\newif\ifarxiv
\arxivtrue
\ifarxiv
\documentclass{article}
\else
\documentclass{hjm1}
\fi

\newif\ifpdflatex
\pdflatexfalse

\usepackage[T1]{fontenc}
\usepackage{amsmath,amsfonts,amssymb,enumitem,url,graphicx,orcidlink}

\DeclareRobustCommand{\VAN}[3]{#2} 

\newcommand\Open{\mathcal O}
\newcommand\real{\mathbb{R}}

\newcommand\nat{\mathbb{N}}
\newcommand\pow{\mathbb{P}}
\newcommand\dc{\mathop{\downarrow}}
\newcommand\upc{\mathop{\uparrow}}
\newcommand\limp{\mathrel{\Rightarrow}}

\newcommand\diff{\smallsetminus}

\newcommand\Rp{\real_+}
\newcommand\creal{\overline{\real}_+}

\newcommand\Lform{{\mathcal L}}
\newcommand{\interior}[1]{int ({#1})} 

\newcommand\dsup{\sup\nolimits^{\scriptstyle\uparrow}}

\newcommand\dcup{\bigcup\nolimits^{\scriptstyle\uparrow}}

\newcommand\eqdef{\mathrel{\buildrel \text{def}\over=}}
\newcommand\Val{\mathbf V}
\newcommand\Smyth{\mathcal Q}

\newcommand\Smythz{\Smyth_0}
\newcommand\SV{\Smyth_\Vt}
\newcommand\SVz{\Smyth_{0\Vt}}

\newcommand\cvx{{\mathrm{cvx}}}
\newcommand\SVc{\SV^\cvx}
\newcommand\Hoare{\mathcal H}
\newcommand\Hoarez{\Hoare_{0}}

\newcommand\HV{\Hoare_\Vt}
\newcommand\HVc{\HV^\cvx}
\newcommand\HVz{\Hoare_{0\Vt}}

\newcommand\QE{{\mathtt Q}}
\newcommand\CF{{\mathtt H}}

\newcommand\Pred{\mathbf{P}}

\newcommand\Angel{{\mathtt{A}}}
\newcommand\Demon{{\mathtt{D}}}
\newcommand\Nature{{\mathtt{P}}}
\newcommand\AN{{\Angel\Nature}}
\newcommand\DN{{\Demon\Nature}}

\newcommand\sub{{\mathrm{sub}}}
\newcommand\super{{\mathrm{super}}}

\newcommand\lin{{\mathrm{lin}}}

\newcommand\one{{\mathbf 1}}
\newcommand\B{\mathfrak{B}}
\newcommand\C{\mathfrak{C}}

\newcommand\rast{\circledast} 
\newcommand\supP{\sup}
\newcommand\qusP{\mathop{\mathrm{qus}}}
\newcommand\minP{\min}
\newcommand\nimP{\mathop{\mathrm{nim}}}

\newcommand\pp{{\perp\perp}}
\newcommand\Cc{\HV^\pp} 
\newcommand\CPred{\Cc\Pred} 
\newcommand\Qq{\SV^\pp} 
\newcommand\QPred{\Qq\Pred} 

\newcommand\shd[2]{\mathop{\mathrm{shd}}^{#1}_{#2}}
\newcommand\shadow[2]{%
  \hbox{%
    \sbox0{$#1B$}%
    \sbox2{$#1>$}%
    \raisebox{\dimexpr\dp2+(\ht0-\ht2)/2}{%
      $#1\mathbin{>}\mspace{-15mu}\mathrel{(} #2$%
    }%
  }%
}
\newcommand\Lhom{\Lform_{\mathrm{hom}}}
\newcommand\Vt{{\mathit{v}}}
\newcommand\Img{\mathrm{Im}\,}
\newcommand{\identity}[1]{\mathrm{id}_{#1}}
\newcommand\conv{\mathop{\mathrm{conv}}}
\newcommand\clconv{\mathop{\overline{\conv}}}

\begin{document}

\renewcommand{\bf}{\bfseries}
\renewcommand{\sc}{\scshape}

\ifarxiv
\title{Just Previsions}
\author{Jean Goubault-Larrecq\\
  \small Universit\'e Paris-Saclay, CNRS, ENS Paris-Saclay,\\
  \small Laboratoire M\'ethodes Formelles, 91190, Gif-sur-Yvette, France}
\date{\relax}
\fi

\newcommand\mykeywords{Prevision, double hyperspace construction, continuous valuation.}
\ifarxiv
\relax
\fi



\newcommand\myabstract{%
  Previsions are positively homogeneous functionals, and are
  generalized forms of integration functionals.  We investigate
  previsions---just previsions, not sublinear or superlinear
  previsions as in previous work.  We show that every prevision can be
  expressed as an infimum of sublinear previsions, and as a supremum
  of superlinear previsions under mild conditions.  This extends to
  homeomorphisms between spaces of previsions and certain hyperspaces
  over spaces of sublinear or superlinear previsions, which can also
  be characterized in terms of orthogonality relations, making the
  construction a variant of a double powerspace construction.%
}

\ifarxiv
\maketitle
\fi

\begin{abstract}
  \myabstract
  \ifarxiv%

  \fi%
\end{abstract}

\ifarxiv
\relax
\fi


\ifarxiv
\newtheorem{thm}{Theorem}[section]
\newtheorem{lem}[thm]{Lemma}
\newtheorem{prop}[thm]{Proposition}
\newcommand\qed{\hfill $\Box$}
\newenvironment{proof}{\emph{Proof.}}{\qed}
\fi
\newtheorem{fact}[thm]{Fact}
\newtheorem{corollary}[thm]{Corollary}
\newtheorem{lemdef}[thm]{Lemma and Definition}

\ifarxiv
\relax
\fi
\newtheorem{definition}[thm]{Definition}

\newtheorem{remark}[thm]{Remark}
\newtheorem{example}[thm]{Example}
\numberwithin{equation}{section}



\noindent
\begin{minipage}{0.25\linewidth}
  \ifarxiv
  \includegraphics[scale=0.2]{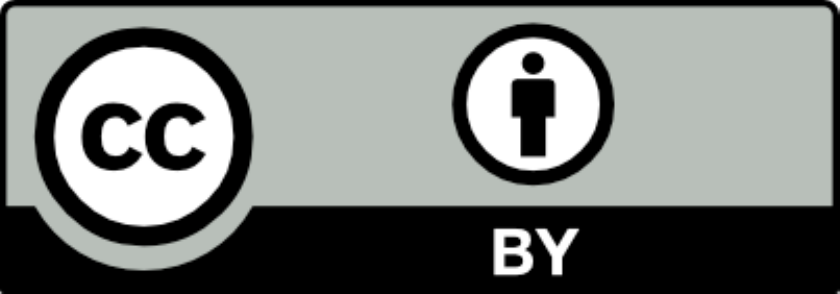}
  \fi
\end{minipage}
\begin{minipage}{0.74\linewidth}
  \scriptsize
  For the purpose of Open Access, a CC-BY public copyright licence has
  been applied by the authors to the present document and will be
  applied to all subsequent versions up to the Author Accepted
  Manuscript arising from this submission.
\end{minipage}

\section{Introduction}
\label{sec:intro}

Let $\creal$ be the space of extended non-negative real numbers
$\Rp \cup \{\infty\}$.  Given a topological space $X$, we write
$\Lform X$ for the set of all lower semicontinuous maps from $X$ to
$\creal$.  A \emph{prevision} is a positively homogenous,
Scott-continuous map from $\Lform X$ to $\creal$ (we will define all
missing terms in Section~\ref{sec:prelim}).  A \emph{linear} prevision
can be thought of as a functional
$h \in \Lform X \mapsto \int h \,d\mu$ for some measure $\mu$ on $X$,
or more precisely for some continuous valuation $\mu$ on $X$, and
previsions that satisfy the more lenient constraints of being
sublinear or superlinear were advocated as models of mixtures of
non-deterministic and probabilistic choice in \cite{Gou-csl07}.

The purpose of this paper is to examine what can be said about
previsions: just previsions, with no sublinearity, superlinearity of
linearity involved.  We will see that every prevision is a pointwise
infimum of sublinear previsions, and also a pointwise supremum of
superlinear previsions under mild assumptions.  This extends to
homeomorphisms between \emph{spaces} of previsions and certain
hyperspaces over spaces of sublinear or superlinear previsions.  We
will obtain those homeomorphisms as restrictions of retractions of
well-known convex hyperspaces over spaces of sublinear or superlinear
previsions.  However, the hyperspaces that are homeomorphic to our
spaces of previsions cannot be defined by convexity properties alone,
and we elucidate what they are (partially) using orthogonality
relations.  We will argue that this is related to the \emph{double
  hyperspace} construction \cite{dBK:comm}.

\emph{Outline.}  Section~\ref{sec:prelim} is a list of preliminaries.
It is a bit long, but I am afraid this is unavoidable.  We give a few
easy properties of spaces of previsions in
Section~\ref{sec:struct-prop-spac}.
Section~\ref{sec:cast-shad-constr} is technical, and introduces the
\emph{cast shadow} construction, which we will need later to study
subnormalized and normalized previsions.  We explain how the double
hyperspace construction arises from orthogonality relations and double
orthogonals in Section~\ref{sec:double-hypersp-orth}.
Section~\ref{sec:PX=QPAPX} establishes our first isomorphism theorem
$\Pred^\rast X \cong \QPred^\rast_\sub X$, as well as a related
construction of $\Pred^\rast X$ as a retract of the convex Smyth
hyperspace $\SVc {\Pred^\rast_\sub X}$.  In passing, we show that
every prevision is a pointwise infimum of sublinear previsions, for
any topological space $X$.  Section~\ref{sec:PX=CPDPX} is very
similar, leads to the second isomorphism theorem
$\Pred^\rast X \cong \CPred^\rast_\super X$ and to the fact that every
prevision is a pointwise supremum of superlinear previsions; but we
will require some conditions on $X$ for that to hold.  We conclude in
Section~\ref{sec:back-orth-relat}, where we close the loop and show
that $\QPred^\rast_\sub X$ and $\CPred^\rast_\super X$ are the double
orthogonals that arise from the construction of
Section~\ref{sec:double-hypersp-orth}.

\section{Preliminaries}
\label{sec:prelim}

\subsection{Topology and order.}  For background on topology, we refer
to \cite{JGL-topology}.  We write $\interior A$ for the interior of
$A$, $cl (A)$ 
for the closure of $A$, 
and $\Open X$ for the lattice of open subsets of $X$.  The
specialization preordering $\leq$ of a topological space $X$ is
defined on points $x, y \in X$ by $x \leq y$ if and only if every open
neighborhood of $x$ contains $y$.  We write $\dc$ for downward closure
and $\upc$ for upward closure with respect to $\leq$.  The closure of
a point $x$ is $\dc x$.

We will also say that $x$ is \emph{below} $y$ and that $y$ is
\emph{above} $x$ when $x \leq y$.  A space is $T_0$ if and only if
$\leq$ is antisymmetric, $T_1$ if and only if $\leq$ is the equality
relation.

A family $D$ of elements of a preordered set $P$ is \emph{directed} if
and only if it is non-empty and every pair of elements of $D$ has an
upper bound in $D$.  In case $P$ is a poset (a partially ordered set),
we write $\dsup D$, or $\dsup_{i \in I} x_i$ when
$D = {(x_i)}_{i \in I}$, for the supremum of a directed family $D$, if
it exists; similarly, we write $\dcup_{i \in I} U_i$ for the union of
a directed family of subsets $U_i$ of a fixed set.

A function $f \colon P \to Q$ between posets is \emph{monotonic} if
and only if for all $x, x' \in P$, $x \leq x'$ implies
$f (x) \leq f (x')$.  It is \emph{Scott-continuous} if and only if $f$
is monotonic and for every directed family ${(x_i)}_{i \in I}$ with a
supremum $x$ in $P$, the (necessarily directed) family of elements
$f (x_i)$ has $f (x)$ as supremum.  Scott-continuity is equivalent to
continuity with the respective Scott topologies on $P$ and $Q$.  The
\emph{Scott topology} on a poset $P$ consists of those subsets
$U$---the \emph{Scott-open subsets} of $P$---that are upwards closed
($x \in U$ and $x \leq x'$ implies $x' \in U$) and such that every
directed family $D$ that has a supremum in $U$ intersects $U$.  That
is most useful in the context of \emph{dcpo}s (short for
directed-complete posets), namely posets in which every directed
family has a supremum.

We see $\creal$ as a complete lattice under the usual ordering $\leq$.
A function $f \colon X \to \creal$ that is continuous from $X$ to
$\creal$ with its Scott topology is known as a \emph{lower
  semicontinuous} function.  The set $\Lform X$ of lower
semicontinuous functions from $X$ to $\creal$ will be ordered
pointwise.  When seen as a topological space, we will always assume
that $\Lform X$ has the Scott topology of that ordering.

\subsection{Hyperspaces.}
The closed subsets of a topological space $X$ form its \emph{Hoare
  hyperspace} $\Hoarez X$.  We order it by inclusion.  The \emph{lower
  Vietoris topology} on $\Hoarez X$ has a subbase of open sets of the
form $\Diamond U \eqdef \{C \in \Hoarez X \mid C \cap U \neq
\emptyset\}$, where $U$ ranges over $\Open X$.  The specialization
ordering of that topology is inclusion, and we will write $\HVz X$ for
$\Hoarez X$ with the lower Vietoris topology.  The set $\Hoare X$ is
equal to $\Hoarez X$ minus the empty set.  With the subspace topology
induced by the inclusion in $\HVz X$, which we also call the lower
Vietoris topology, we obtain a space $\HV X$, with the same subbase of
open sets.

A subset $Q$ of a space $X$ is compact if and only if every open cover
of $Q$ has a finite subcover; no separation is required.  Alexander's
subbase lemma states that given any subbase $\mathcal B$, of the
topology of $X$, a subset $Q$ of $X$ is compact if and only if every
open cover of $Q$ by elements of $\mathcal B$ has a finite subcover.
A space $X$ is \emph{locally compact} if and only if every point has a
base of compact neighborhoods.  There is a weaker property,
\emph{core-compactness}, which we will sometimes refer to, but which
we will not define here \cite[Section~5.2.1]{JGL-topology}.

A saturated subset is one that is equal to the intersection of its
open neighborhoods, or equivalently, that is upward-closed with
respect to the specialization preordering $\leq$.  The compact
saturated subsets of $X$ form its \emph{Smyth hyperspace} $\Smythz X$.
We let $\Smyth X$ be $\Smythz X$ minus the empty set.  Both are
ordered by reverse inclusion, which is also the specialization
ordering of the \emph{upper Vietoris topology}.  The latter has a base
of open sets of the form
$\Box U \eqdef \{Q \in \Smythz X \text{ (resp.  $\Smyth X$)}\mid Q
\subseteq U\}$.

One may compose $\Hoarez$ and $\Smythz$ in two ways, and in most cases
the order of composition does not matter.  Historically, the first
result of this kind is due to Flannery and Martin \cite{FM:QH}.  It
was extended to all dcpos by Heckmann \cite{Heckmann:QH}, and to all
locales by Vickers and Townsend \cite{VT:double:power}, with
definitions of the respective Hoare and Smyth powerdomains or
powerlocales that are close, although not identical.  In the setting
of topological spaces, the relevant results are due to de Brecht and
Kawai \cite{dBK:comm}.  For every topological space $X$, there is a
map $\sigma_X \colon \HVz {\SVz X} \to \SVz {\HVz X}$, defined by:
\begin{align}
  \label{eq:sigma}
  \sigma_X (\mathcal C)
  & \eqdef \{C \in \HVz X \mid \forall Q \in \mathcal C, C \cap Q \neq \emptyset\},
\end{align}
and it is a topological embedding.  It is a homeomorphism if and only
if $X$ is \emph{consonant} \cite[Theorem~6.13]{dBK:comm}, namely if
and only if the Scott topology on $\Open X$ has a subbase of open
subsets of the form
$\blacksquare Q \eqdef \{U \in \Open X \mid Q \subseteq U\}$, where
$Q$ ranges over $\Smythz X$.  The inverse is the map
$\tau_X \colon \SVz {\HVz X} \to \HVz {\SVz X}$, defined by:
\begin{align}
  \label{eq:tau}
  \tau_X (\mathcal Q)
  & \eqdef \{Q \in \SVz X \mid \forall C
  \in \mathcal Q, C \cap Q \neq \emptyset\}.
\end{align}

\subsection{Continuous valuations.}
A \emph{continuous valuation} on a space $X$ is a map
$\nu \colon \Open X \to \creal$ that is \emph{strict}
($\nu (\emptyset)=0$), \emph{modular} (for all $U, V \in \Open X$,
$\nu (U) + \nu (V) = \nu (U \cup V) + \nu (U \cap V)$) and
Scott-continuous.  We say that $\nu$ is 
a \emph{probability} valuation if and only if $\nu (X)=1$, and a
\emph{subprobability} valuation if and only if $\nu (X) \leq 1$.
Continuous valuations are an alternative to measures that have become
popular in domain theory \cite{jones89,Jones:proba}.  The first
results that connected continuous valuations and measures are due to
Saheb-Djahromi \cite{saheb-djahromi:meas} and Lawson
\cite{Lawson:valuation}.  The most informative results on this matter
are the following.  In one direction, every measure on the Borel
$\sigma$-algebra of $X$ induces a continuous valuation on $X$ by
restriction to the open sets, if $X$ is hereditarily Lindel\"of
(namely, if every directed family of open sets contains a cofinal
monotone sequence).  This is an easy observation, and one half of
Adamski's theorem \cite[Theorem~3.1]{Adamski:measures}, which states
that a space is hereditary Lindel\"of if and only if every measure on
its Borel $\sigma$-algebra restricts to a continuous valuation on its
open sets.  By definition, a \emph{$\tau$-smooth} measure is one whose
restriction to the open sets is a continuous valuation: hence all
measures on hereditarily Lindel\"of spaces are $\tau$-smooth.  In the
other direction, every continuous valuation on a space $X$ extends to
a measure on the Borel sets provided that $X$ is an LCS-complete space
\cite[Theorem~1]{dBGLJL:LCS}, a class of spaces that includes the
locally compact sober spaces, the quasi-Polish spaces and therefore
the locally compact Hausdorff spaces and the Polish spaces.

Let $\Val X$ denote the space of continuous valuations on a space $X$,
with the following \emph{weak topology}.  Its subbasic open sets are
$[U > r] \eqdef \{\nu \in \Val X \mid \nu (U) > r\}$, where
$U \in \Open X$ and $r \in \Rp$.  We define its subspace $\Val_1 X$ of
probability valuations and $\Val_{\leq 1} X$ (subprobability)
similarly.  The specialization ordering of each is the
\emph{stochastic ordering} $\leq$ given by $\nu \leq \nu'$ if and only
if $\nu (U) \leq \nu' (U)$ for every $U \in \Open X$.

For all $h \in \Lform X$ and $\nu \in \Val X$, the integral
$\int h \,d\nu$ is best defined by the Choquet formula
$\int_0^\infty \nu (h^{-1} (]t, \infty]))\,dt$, using an indefinite
Riemann integral \cite[Section~4]{Tix:bewertung}.  Then
$\int h \,d\nu$ is Scott-continuous in $h$ and in $\nu$; it is also
linear in $h$ and $\nu$, meaning that it preserves finite sums and
products by scalars from $\Rp$.  The weak topology on $\Val X$ (and
its subspaces $\Val_{\leq 1} X$, $\Val_1 X$) has an alternative
subbase of open sets of the form
$[h > r] \eqdef \{\nu \in \Val X \mid \int h \,d\nu > r\}$, where
$h \in \Lform X$ and $r \in \Rp$; this was originally observed by Jung
\cite[Theorem~3.3]{Jung:scs:prob} for $\Val_{\leq 1} X$ and
$\Val_1 X$.

\subsection{Previsions.}
\emph{Previsions} are an elaboration on Walley's notion of prevision
in economics \cite{Walley:prev}.  We will borrow most of what we need
from \cite{JGL-mscs16}, see also the errata \cite{JGL:iso:err}.  A
\emph{prevision} on a space $X$ is a Scott-continuous map
$F \colon \Lform X \to \creal$ that is \emph{positively homogeneous}
in the sense that $F (ah)=aF(h)$ for all $a \in \Rp$ and
$h \in \Lform X$.  (The notation $ah$ makes sense even if $a=0$, with
the convention that $0.\infty = 0$.)  There is a space $\Pred X$ of
previsions on $X$, whose topology is generated by sets
$[h > r] \eqdef \{F \mid F (h) > r\}$, $h \in \Lform X$, $r \in \Rp$.

Every continuous valuation $\nu$ on $X$ gives rise to a prevision
$\Lambda \colon h \mapsto \int h \,d\nu$.  Such a prevision is
\emph{linear}, in the sense that (it is positively homogeneous and)
$\Lambda (h+h') = \Lambda (h) + \Lambda (h')$ for all
$h, h' \in \Lform X$.  Let $\Pred_\lin X$ be the subspace of $\Pred X$
of linear previsions.  Conversely, every linear prevision
$\Lambda \in \Pred_\lin X$ gives rise to a continuous valuation
$U \mapsto \Lambda (\chi_U)$, where $\chi_U$ is the characteristic map of
the open set $U$, and the two constructions are inverse of each other.
Additionally, those two constructions define continuous maps between
$\Val X$ and $\Pred_\lin X$ \cite[Satz~4.16]{Tix:bewertung}.  We will
therefore equate continuous valuations with linear previsions.

A prevision $F$ is \emph{sublinear} (resp., \emph{superlinear}) if and
only if $F (h+h') \leq F (h) + F (h')$ (resp., $\geq$) for all
$h, h' \in \Lform X$.  We will write $\Pred_\sub X$ (resp.\
$\Pred_\super X$) for the subspace of $\Pred X$ consisting of
sublinear (resp.\ superlinear) previsions.  This were written as
$\Pred_{\AN} X$ (resp.\ $\Pred_{\DN} X$) in \cite{JGL-mscs16}.

Among the continuous valuations, there are the (sub)probability
valuations.  Similarly, we say that a prevision $F$ is
\emph{subnormalized} (resp., \emph{normalized}) iff
$F (\one+h) \leq \one+F (h)$ (resp., $=$) for every $h \in \Lform X$,
where $\one$ is the constant function with value $1$.  If $F$ is
subnormalized (resp.\ normalized), then $F (\one) \leq 1$ (resp.\
$=$); the converse holds if $F$ is linear.  The homeomorphism
$\Val X \cong \Pred_\lin X$ restricts to homeomorphisms between
$\Val_{\leq 1} X$ (resp., $\Val_1 X$) and the subspace
$\Pred_\lin^{\leq 1} X$ (resp., $\Pred_\lin^1 X$) of subnormalized
(resp., normalized) linear previsions on $X$.  We write
$\Pred^{\leq 1} X$, $\Pred_\sub^{\leq 1} X$,
$\Pred_\super^{\leq 1} X$, $\Pred^1 X$, $\Pred_\super^{\leq 1} X$,
$\Pred_\super^1 X$ for the corresponding spaces of (sub)normalized,
plain/sublinear/superlinear previsions.  For short, we write
$\Pred^\rast X$, $\Pred_\sub^\rast X$, $\Pred_\super^\rast X$,
$\Pred_\lin^\rast X$ where $\rast$ can be nothing, ``$\leq 1$'', or
``$1$''.

\subsection{Cones.}  Spaces such as $\Lform X$, $\Val X$ or $\Pred X$ have
additional algebraic structure: they are cones, in the following
sense.  A \emph{cone} is a commutative monoid $(\C, +, 0)$ (or just
$\C$ for short) with an action $a, x \mapsto a \cdot x$ of the
semi-ring $\Rp$ on $\C$, that is:
\begin{equation}
  \label{eq:cone}
  \begin{array}{ccc}
    0 \cdot x = 0 & (ab) \cdot x = a \cdot (b \cdot x) & 1 \cdot x=x\\
    a \cdot 0=0 & a \cdot (x+y) = a \cdot x+a \cdot y & (a+b) \cdot x = a \cdot x+b \cdot x
  \end{array}
\end{equation}
for all $a, b \in \Rp$ and $x, y \in \C$.  We sometimes write $ax$ for
$a \cdot x$.  In other words, a cone is just like a vector space
except that one cannot take opposites or multiply by negative scalars
in general.  Following Keimel \cite{Keimel:topcones2}, a
\emph{semitopological} cone is a cone $\C$ with a topology that makes
both $+$ and $\cdot$ separately continuous, where $\Rp$ has the Scott
topology of $\leq$.  Scalar multiplication is then automatically
jointly continuous \cite[Corollary~6.9~(c)]{Keimel:topcones2}.  A
\emph{topological} cone is a semitopological cone in which $+$ (and
therefore $\cdot$, too) is jointly continuous.  $\Lform X$, with the
Scott topology and pointwise addition and scalar multiplication, is
always a semitopological cone; it is a topological cone if $X$ is
locally compact, or more generally core-compact.  $\Val X$, with the
weak topology and pointwise addition and scalar multiplication, is
always a topological cone \cite[Example~5.11]{Keimel:topcones2}.

We will often write $x +_\alpha y$ for the convex combination
$\alpha \cdot x + (1-\alpha) \cdot y$.  A subset $A$ of a cone $\C$ is
\emph{convex} if and only if it contains every convex combination
$x +_\alpha y$ of points $x, y \in A$, where $\alpha \in [0, 1]$.
Every intersection of convex sets is convex, every directed union of
convex sets is convex.  A map is \emph{affine} if and only if it
preserves convex combinations.  Inverse and direct images of convex
sets under affine maps are convex.

The smallest convex subset containing $A$ is its \emph{convex hull}
$\conv A$.  This is also the set of barycenters
$\sum_{i=1}^n a_i \cdot x_i$ where $n \geq 1$ and $\vec a$ ranges over
$\Delta_n$.  We write $\Delta_n$ for the simplex
$\{\vec a \in (\Rp)^n \mid \sum_{i=1}^n a_i = 1\}$, and we use the
convention that $\vec a$ abbreviates $(a_1, \cdots, a_n)$.

In a semitopological cone, the closure of a convex set is convex
\cite[Lemma 4.10]{Keimel:topcones2}.  It follows that the \emph{closed
convex hull} $\clconv A$ of any subset $A$ is equal to the closure $cl
(\conv A)$ of the convex hull of $A$.

Given a semitopological cone $\C$, a function $F \colon \C \to \creal$
is \emph{positively homogeneous} if and only if $F (a \cdot x) = a F
(x)$ for all $a \in \Rp$ and $x \in \C$.  Hence previsions on $X$ are
the same thing as positively homogeneous, lower semicontinuous
functions from $\Lform X$ to $\creal$.
$F$ is \emph{superlinear} (resp.\ \emph{sublinear}) if and only if it
is positively homogeneous and satisfies $F (x+x') \geq F (x) + F (x')$
(resp.\ $\leq$) for all $x, x' \in \C$.

The \emph{upper Minkowski functional} of a subset $A$ of a cone $\C$,
which Keimel writes as $F^A$, and which we prefer to write as $M^A$,
maps every $x \in C$ to $\sup \{r \in \Rp \mid x \in r \cdot A\}$,
where the sup is taken in $\creal$, the sup of the empty set is $0$,
and $r \cdot A \eqdef \{r \cdot y \mid y \in A\}$ (see
\cite[Section~7]{Keimel:topcones2}).  In a semitopological cone $\C$,
there is an order isomorphism between the poset $\Open^* \C$ of proper
open subsets of $\C$, ordered by inclusion, and the poset $\Lhom \C$
of positively homogeneous, lower semicontinuous maps from $\C$ to
$\creal$.  In one direction, every positively homogeneous, lower
semicontinuous map $F \colon \C \to \creal$ defines a proper open
subset $F^{-1} (]1, \infty])$ of $\C$.  The inverse operation sends
every proper open subset $U$ of $\C$ to $M^U$.  Additionally, $M^U$ is
superlinear if and only if $U$ is convex (and proper), and sublinear
if and only if $\C \diff U$ is convex (and non-empty).

Keimel's \emph{sandwich theorem} \cite[Theorem~8.2]{Keimel:topcones2}
states that in a semitopological cone $C$, given a lower
semicontinuous superlinear map $q \colon C \to \creal$ and a sublinear
map $p \colon C \to \creal$ 
such that $q \leq p$, there is a lower semicontinuous linear map
$\Lambda \colon C \to \creal$ such that $q \leq \Lambda \leq p$.
(Functions are compared pointwise.)

\subsection{Representation theorems for sublinear and superlinear
  previsions.}  For every semitopological cone $\C$, or more generally
for every convex subspace $\B$ of a semitopological cone, we write
$\SVc \B$ for the subspace of $\SV \B$ consisting of non-empty convex
compact saturated subsets of $\B$ and $\HVc \B$ for the subspace of
$\HV \B$ consisting of non-empty convex closed subsets of $\B$.  When
$\B$ is a semitopological cone $\C$, those are the convex upper and
lower powercones of Keimel, Tix and Plotkin
\cite[Section~4]{TKP:nondet:prob}.  Writing $\Box^\cvx U$ for
$\Box U \cap \SVc \B$, the sets $\Box^\cvx U$ with $U \in \Open \B$
form a base of the topology of $\SVc \B$.  Similarly, the sets
$\Diamond^\cvx U \eqdef \Diamond U \cap \HVc \B$ form a subbase of the
topology of $\HVc \B$.

Let $\rast$ be nothing, ``$\leq 1$'' or ``$1$''.  There is a
continuous function
$r_{\super} \colon \SV {\Pred^\rast_{\lin} X} \to \Pred^\rast_{\super}
X$, which maps $\mathcal Q$ to
$h \in \Lform X \mapsto \min_{\Lambda \in \mathcal Q} \Lambda (h)$,
and a continuous function
$s^\rast_\super \colon F \mapsto \{\Lambda \in \Pred^\rast_\lin X \mid
\Lambda \geq F\}$ in the other direction such that
$r_{\super} \circ s^\rast_\super = \identity {\Pred^\rast_{\super}
  X}$.  
In other words, $\Pred^\rast_{\super} X$ is
a \emph{retract} of $\SV {\Pred^\rast_{\lin} X}$ \cite[Proposition
3.22]{JGL-mscs16}.  (Those functions were written as $r_\DN$ and
$s^\rast_\DN$ in \cite{JGL-mscs16}.)  This retraction cuts down to a
homeomorphism between $\SVc {\Pred^\rast_{\lin} X}$ and
$\Pred^\rast_{\super} X$ \cite[Theorem~4.15]{JGL-mscs16}.  Remarkably,
this holds without any condition on $X$.

The situation with Hoare hyperspaces and sublinear previsions is
similar but a bit more complex.  There is a continuous function
$r_{\sub} \colon \HV {\Pred^\rast_{\lin} X} \to \Pred^\rast_{\sub} X$,
which maps $\mathcal C$ to
$h \in \Lform X \mapsto \sup_{\Lambda \in \mathcal C} \Lambda (h)$,
and a function
$s^\rast_\sub \colon F \mapsto \{\Lambda \in \Pred^\rast_\lin X \mid
\Lambda \leq F\}$ in the other direction.  Those form a retraction of
$\HV {\Pred^\rast_{\lin} X}$ onto $\Pred^\rast_\sub X$ provided that
$X$ is $\AN_\rast$-friendly
\cite[Proposition~3.11]{JGL-mscs16,JGL:iso:err}.  We will leave out
the complex and relatively ad hoc definition of $\AN_\rast$-friendly
spaces \cite[Definition~1]{JGL:iso:err}, but we will note that every
locally compact (even just core-compact) space is
$\AN_\rast$-friendly, and that every LCS-complete space is
$\AN_\rast$-friendly \cite[end of Section~2]{JGL:iso:err}.

Finally, the $r_\sub, s^\rast_\sub$ retraction restricts to a
homeomorphism between $\HVc {\Pred^\rast_{\lin} X}$ and
$\Pred^\rast_{\sub} X$ \cite[Theorem~4.11]{JGL-mscs16,JGL:iso:err},
provided that $X$ is $\AN_\rast$-friendly.

\section{Structural properties of spaces of previsions}
\label{sec:struct-prop-spac}

We collect a few easy properties of spaces of previsions, for future use.
\begin{prop}
  \label{prop:PX:cone}
  For every topological space, $\Pred X$, $\Pred_\sub X$,
  $\Pred_\super X$ and $\Pred_\lin X$ are topological cones, with the
  zero prevision as $0$, and pointwise addition and scalar
  multiplication.
\end{prop}
\begin{proof}
  Let $Z$ be any of the spaces of previous listed above.  Let
  $h \in \Lform X$ and $r \in \Rp$.  For every $a \in \Rp$,
  $(a \cdot \_)^{-1} ([h > r]) = [ah > r]$.  For every $P \in Z$,
  $(\_ \cdot P)^{-1} ([h > r]) = ]r/P (h), \infty]$ if
  $P (h) \neq 0, \infty$, empty if $P (h)=0$, $\Rp$ if $P (h)=\infty$.
  Hence scalar multiplication is separately continuous.  Finally,
  $+^{-1} ([h > r]) = \bigcup_{a, b \in \Rp, a+b \geq r} ([h > a]
  \times [h > b])$, so addition is continuous.
\end{proof}

\begin{lem}
  \label{lemma:game}
  Let $X$ be a topological space.  For every subset $A$ of $X$, the
  map $h \mapsto \sup_{x \in A} h (x)$ is a subnormalized sublinear
  prevision, which is normalized if $A$ is non-empty.  For every
  non-empty compact subset $Q$ of $X$, the map
  $h \mapsto \min_{x \in Q} h (x)$ is a normalized superlinear
  prevision.
\end{lem}
\begin{proof}
  The first part is clear.  For the second part, for every
  $h \in \Lform X$, $P (h) \eqdef \min_{x \in Q} h (x)$ makes sense
  because the infimum of a lower semicontinuous map on a non-empty
  compact set is attained.  This also implies that $P (h) > r$ if and
  only if $Q \subseteq h^{-1} (]r, \infty])$.  Positive homogeneity
  and normalization are clear.  We claim that $P$ is Scott-continuous.
  Let ${(h_i)}_{i \in I}$ be a directed family with (pointwise)
  supremum $h$ in $\Lform X$.  $P$ is monotonic, so
  $\dsup_{i \in I} P (h_i) \leq P (h)$.  Conversely, for every
  $r \in \Rp$ such that $r < P (h)$, $Q \subseteq [h > r]$, and
  $[h > r] = \dcup_{i \in I} [h_i > r]$; since $Q$ is compact and
  ${([h_i > r])}_{i \in I}$ is directed, there is an $i \in I$ such
  that $Q \subseteq [h_i > r]$, whence $r < P (h_i)$.
\end{proof}

\begin{prop}
  \label{prop:PX:sup}
  Let $\rast$ be nothing, ``$\leq 1$'' or ``$1$'' and $X$ be a
  topological space.
  \begin{enumerate}
  \item The specialization ordering of $\Pred^\rast X$, of
    $\Pred^\rast_\sub X$, of $\Pred^\rast_\super X$, of
    $\Pred^\rast_\lin X$ is the pointwise ordering $\leq$.
  \item In $\Pred^\rast X$ or in $\Pred^\rast_\sub X$, every
    (non-empty if $\rast$ is ``$1$'') family has a supremum, which is
    computed pointwise.
  \item If $\rast$ is nothing or ``$\leq 1$'', then $0$ is the least
    element of $\Pred^\rast X$, of $\Pred^\rast_\sub X$, of
    $\Pred^\rast_\super X$, of $\Pred^\rast_\lin X$.  If $\rast$ is
    ``$1$'' and $X$ is compact and non-empty, then $\Pred^\rast X$ and
    $\Pred^\rast_\super X$ have a least element, which maps $h \in
    \Lform X$ to $\min_{x \in X} h (x)$.
  \item In $\Pred^\rast X$ or in $\Pred^\rast_\super X$, every
    non-empty finite family has an infimum, which is computed
    pointwise.
  \end{enumerate}
\end{prop}
\begin{proof}
  (1) If $P$ is below $P'$ in the pointwise ordering, every subbasic
  open set $[h > r]$ that contains $P$ also contains $P'$.  If $P$ is
  below $P'$ in the specialization ordering, then for all
  $h \in \Lform X$ and $r \in \Rp$, $P \in [h > r]$ implies
  $P' \in [h > r]$.  Hence $P (h) > r$ implies $P' (h) > r$.  Taking
  suprema over $r$ yields $P (h) \leq P' (h)$.

  (2) Every pointwise supremum of lower semicontinuous maps is lower
  semicontinuous, and it is easy to see that pointwise suprema
  preserves positive homogeneity and subnormalization.  Normalization
  ($P (\one+h) = 1+P(h)$) is also preserved provided the supremum is
  over a non-empty family.  If ${(F_i)}_{i \in I}$ is any family of
  sublinear previsions, then for all $h', h' \in \Lform X$, $\sup_{i
    \in I} F_i (h+h') \leq \sup_{i \in I} F_i (h) + F_i (h') \leq
  \sup_{i \in I} F_i (h) + \sup_{i \in I} F_i (h')$, so $\sup_{i \in
    I} F_i$ is sublinear.

  (3) The cases where $\rast$ is nothing or ``$\leq 1$'' are obvious.
  If $\rast$ is ``$1$'' and $X$ is compact and non-empty, then
  $h \mapsto \min_{x \in X} h (x)$ is in $\Pred^1_\super X$ by
  Lemma~\ref{lemma:game}.  For every $P \in \Pred^1 X$, for every
  $h \in \Lform X$,
  $P (h) \geq P (\min_{x \in X} h (x) \cdot \one) = \min_{x \in X}
  \allowbreak h (x) . P (\one) = \min_{x \in X} h (x)$, so
  $h \mapsto \min_{x \in X} h (x)$ is least.

  (4) This is similar to (2), with the exception that only
  \emph{finite} pointwise infima of lower semicontinuous maps are
  lower semicontinuous.  We restrict to non-empty families: the
  pointwise infimum of the empty family would be the constant $\infty$
  map, which fails to be a prevision since it does not map $0$ to $0$.
\end{proof}

\section{The cast shadow construction}
\label{sec:cast-shad-constr}

There are some delicate, technical cone-theoretic constructions we
will need to do in order to study $\Pred^\rast X$ when $\rast$ is
``$\leq 1$'' or ``$1$''.  Introducing them when we need them would
disrupt the flow of exposition too much, and we prefer to get them out
of the way right now.

Geometrically, the \emph{cast shadow} construction $\shadow A {x_0}$
is obtained by placing a source of light at $x_0$, and looking at the
shadow that $A$ casts, including $A$ itself.

\begin{lem}
  \label{lemma:P:shd}
  Let $\C$ be a cone and $x_0 \in \C$.  We define
  $\shd {x_0} \alpha \colon \C \to \C$ by
  $\shd {x_0} \alpha (x) \eqdef x +_\alpha x_0$ for every $x \in \C$
  and every $\alpha \in [0, 1]$.  For every subset $A$ of $\C$, let
  $\shadow {A} {x_0} \eqdef \bigcup_{\beta \in {]0, 1]}} {(\shd {x_0}
    \beta)}^{-1} (A)$.  Then:
  \begin{enumerate}
  \item for all $\alpha, \beta \in [0, 1]$,
    $\shd {x_0} \beta \circ \shd {x_0} \alpha = \shd {x_0}
    {\alpha\beta}$;
  \item if $A$ is convex, then $\shadow {A} {x_0}$ is convex;
  \item if the complement of $A$ is convex, then the complement of
    $\shadow {A} {x_0}$ is convex;
  \item for every $\alpha \in {]0, 1]}$,
    ${(\shd {x_0} \alpha)}^{-1} (\shadow {A} {x_0}) \subseteq \shadow {A}
    {x_0}$;
  \item if $\C$ is semitopological and $A$ is open in $\C$, then
    $\shadow {A} {x_0}$ is open is $\C$.
  \end{enumerate}
\end{lem}
\begin{proof}
  (1)  For every $x \in \C$,
  $\shd {x_0} \beta (\shd {x_0} \alpha (x)) = (x +_\alpha x_0) +_\beta
  x_0 = \beta \cdot (\alpha \cdot x + (1-\alpha) \cdot x_0) +
  (1-\beta) \cdot x_0 = x +_{\alpha\beta} x_0 = \shd {x_0}
    {\alpha\beta} (x)$.

    (2) We note that: $(*)$ for all $a, \beta, \gamma \in [0, 1]$,
    $\beta \gamma \leq a \gamma + (1-a) \beta$.  We prove this when
    $0 < \beta \leq \gamma$, since the remaining case
    $0 < \gamma \leq \beta$ follows by exchanging $\beta$ and $\gamma$
    and replacing $a$ by $1-a$.  We fix $\gamma$, and we define
    $f (\beta) \eqdef a \gamma + (1-a) \beta - \beta\gamma$.  This is
    a linear map, $f (0) = a \gamma \geq 0$,
    $f (\gamma) = \gamma - \gamma^2 = \gamma (1-\gamma) \geq 0$, so
    $f$ is non-negative on the interval $[0, \gamma]$, and therefore
    $(*)$ holds.

    \ifarxiv 
    \noindent
    \begin{minipage}{0.54\linewidth}
      \strut\hskip10pt Let $x, y \in \shadow {A} {x_0}$ and
      $a \in [0, 1]$.  If $a$ is equal to $0$ or to $1$, then
      $x +_a y$ is equal to $y$ or to $x$, hence is in
      $\shadow {A} {x_0}$.  Let us therefore assume that $0 < a < 1$.
      We have $\shd {x_0} \beta (x) \in A$ and
      $\shd {x_0} \gamma (y) \in A$ for some
      $\beta, \gamma \in {]0, 1]}$.  In other words, $x +_\beta x_0$
      and $y +_\gamma x_0$ are in $A$.  The situation is as shown in
      the picture on the right.
    \end{minipage}
    \begin{minipage}{0.45\linewidth}
      \centering
      \includegraphics[scale=0.28]{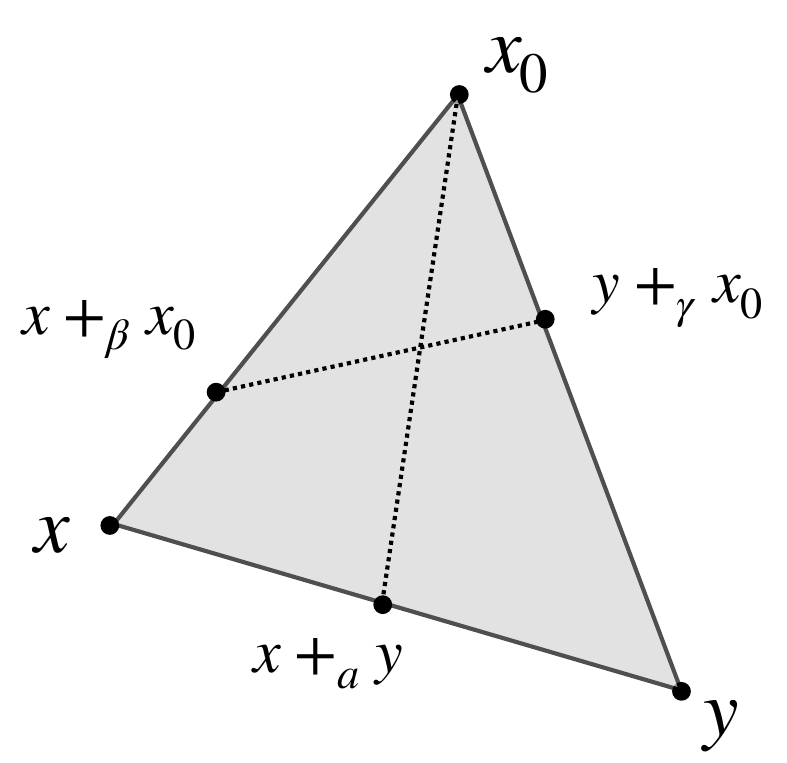}
    \end{minipage}
    \noindent
    We draw the two dotted segments, and we find the point at their
    intersection.  %
    \fi%
    Formally, let
    $\alpha \eqdef \frac {\beta \gamma} {a \gamma + (1-a) \beta}$, and
    let $\alpha' \eqdef \frac {a \gamma} {a \gamma + (1-a) \beta}$.
    It is clear that $\alpha' \in {]0, 1]}$.  Since
    $\beta, \gamma > 0$, we have $\alpha > 0$, and $\alpha \leq 1$ is
    a consequence of $(*)$.  We claim that
    $\shd {x_0} \alpha (x +_a y) = \shd {x_0} \beta (x) +_{\alpha'}
    \shd {x_0} \gamma (y)$:
  \begin{align*}
    \shd {x_0} \alpha\nolimits (x +_a y)
    & = \alpha \cdot (a \cdot x + (1-a) \cdot y) + (1-\alpha) \cdot
      x_0 \\
    & = \alpha a \cdot x + \alpha (1-a) \cdot y + (1-\alpha) \cdot
      x_0 \\
    \shd {x_0} \beta\nolimits (x) +_{\alpha'} \shd {x_0} \gamma\nolimits (y)
    & = \alpha' \cdot (\beta \cdot x + (1-\beta) \cdot x_0)
      + (1-\alpha') \cdot (\gamma \cdot y + (1-\gamma) \cdot x_0) \\
    & = \alpha'\beta \cdot x + (1-\alpha')\gamma \cdot y
      + [\alpha'(1-\beta) + (1-\alpha')(1-\gamma)] \cdot x_0.
  \end{align*}
  In order to see that those two values are equal, we check that
  $\alpha a = \alpha'\beta$, $\alpha (1-a) = (1-\alpha')\gamma$ and
  $1-\alpha = \alpha'(1-\beta) + (1-\alpha')(1-\gamma)$.  For the
  first one,
  $\alpha'\beta = \frac {a \beta \gamma} {a \gamma + (1-a) \beta} =
  \alpha a$.  For the second one,
  $(1-\alpha')\gamma = \frac {(1-a)\beta} {a \gamma + (1-a) \beta}
  \gamma = \frac {(1-a)\beta\gamma} {a \gamma + (1-a) \beta} = (1-a)
  \alpha$.  The third one follows from the first two, since
  $\alpha' (1-\beta) + (1-\alpha') (1-\gamma) = 1 - \alpha' \beta -
  (1-\alpha') \gamma = 1 - \alpha a - \alpha (1-a) = 1 - \alpha$.

  Since $\shd {x_0} \beta (x)$ and $\shd {x_0} \gamma (y)$ are in the
  convex set $A$,
  $\shd {x_0} \beta (x) +_{\alpha'} \shd {x_0} \gamma (y)$ is in $A$, so
  $x +_a y \in {(\shd {x_0} \alpha)}^{-1} (A) \subseteq \shadow {A}
  {x_0}$.  Therefore $\shadow {A} {x_0}$ is convex.

  (3)  For every $\beta \in [0, 1]$, $\shd {x_0} \beta$ is affine, as
  one checks easily.  It follows that for every convex subset $B$ of
  $\C$, ${(\shd {x_0} \beta)}^{-1} (B)$ is convex.

  The complement of $\shadow {A} {x_0}$ is
  $\bigcap_{\beta \in {]0, 1]}} {(\shd {x_0} \beta)}^{-1} (B)$, where
  $B$ is the complement of $A$.  We have seen that, for each
  $\beta \in {]0, 1]}$, ${(\shd {x_0} \beta)}^{-1} (B)$ is convex, and
  intersections of convex sets are easily seen to be convex.
  
  (4)  For every $\alpha \in {]0, 1]}$,
  ${(\shd {x_0} \alpha)}^{-1} (\shadow {A} {x_0}) = \bigcup_{\beta \in
    {]0, 1]}} {(\shd {x_0} \alpha)}^{-1} ({(\shd {x_0} \beta)}^{-1} (
  A))$ is equal to
  $\bigcup_{\beta \in {]0, 1]}} {(\shd {x_0} {\alpha\beta})}^{-1} (A)$
  by item~(1), hence is included in $\shadow {A} {x_0}$.

  (5)  For every $\beta \in [0, 1]$, $\shd {x_0} \beta$ is continuous,
  so if $A$ is open, then ${(\shd {x_0} \alpha)}^{-1} (A)$ is open,
  too, and we conclude since unions of open sets are open.
\end{proof}

\begin{lem}
  \label{lemma:P:subnorm}
  Let $X$ be a topological space, $P \in \Pred X$, and
  $\mathcal V \eqdef P^{-1} (]1, \infty])$.  Then:
  \begin{enumerate}
  \item $P$ is subnormalized if and only if
    ${(\shd \one \alpha)}^{-1} (\mathcal V) \subseteq \mathcal V$ for
    every $\alpha \in {]0, 1]}$;
  \item if $P$ is subnormalized, then for every convex open subset
    $\mathcal U$ of $\Lform X$ included in $\mathcal V$, there is a
    subnormalized superlinear prevision $F$ such that $F \leq P$ and
    $\mathcal U \subseteq F^{-1} (]1, \infty])$.
  \end{enumerate}
\end{lem}
\begin{proof}
  (1) Let us assume that $P$ is subnormalized.  For every
  $\alpha \in {]0, 1]}$, for every
  $h \in {(\shd \one \alpha)}^{-1} (\mathcal V)$, we have
  $P (\shd \one \alpha (h)) > 1$, namely
  $P (\alpha \cdot h + (1-\alpha) \cdot \one) > 1$.  Since $P$ is
  subnormalized, the left-hand side is smaller than or equal to
  $\alpha P (h) + (1-\alpha)$.  Hence $\alpha P (h) > \alpha$.  Since
  $\alpha > 0$, it follows that $P (h) > 1$, namely that
  $h \in \mathcal V$.  Hence
  ${(\shd \one \alpha)}^{-1} (\mathcal V) \subseteq \mathcal V$.

  Conversely, let us assume that
  ${(\shd \one \alpha)}^{-1} (\mathcal V) \subseteq \mathcal V$ for
  every $\alpha \in {]0, 1]}$.  Let $h \in \Lform X$.  In order to
  show that $P (\one+h) \leq 1+P (h)$, we consider any
  $r \in \Rp \diff \{0\}$ such that $r < P (\one+h)$, and we claim
  that $r \leq 1+P(h)$.  If $r \leq 1$, this is obvious, so we assume
  $r > 1$.  Since $r < P (\one+h)$,
  $\frac 1 r \cdot (\one+h) \in \mathcal V$.  Let
  $\alpha \eqdef \frac {r-1} r$.  Then
  $\shd \one \alpha (\frac 1 {r-1} \cdot h) = \frac 1 r \cdot h +
  \frac 1 r \cdot \one = \frac 1 r \cdot (\one+h) \in \mathcal V$.
  Therefore
  $\frac 1 {r-1} \cdot h \in {(\shd \one \alpha)}^{-1} (\mathcal V)$.
  By assumption, $\frac 1 {r-1} \cdot h \in \mathcal V$, namely 
  $P (h) > r-1$, whence $r < 1+P(h)$.

  (2) We let $F$ be the upper Minkowski functional of
  $\shadow {\mathcal U} {\one}$.  By Lemma~\ref{lemma:P:shd} (2)
  and~(5), $\shadow {\mathcal U} {\one}$ is convex and open.  It is
  proper because it does not contain the zero function $0$: otherwise
  $\shd \one \alpha (0) = (1-\alpha) \cdot \one$ would be in
  $\mathcal U$ for some $\alpha \in {]0, 1]}$, hence in $\mathcal V$,
  so $(1-\alpha) P (\one) > 1$, which is impossible since
  $P (\one) \leq 1$.  Therefore $F = M^{\shadow {\mathcal U} {\one}}$
  is a superlinear prevision.  Now
  $F^{-1} (]1, \infty]) = \shadow {\mathcal U} {\one} = \bigcup_{\beta
    \in {]0, 1]}} {(\shd \one \beta)}^{-1} (\mathcal U) \subseteq
  \bigcup_{\beta \in {]0, 1]}} {(\shd \one \beta)}^{-1} (\mathcal V)$,
  which is included in $\mathcal V = P^{-1} (]1, \infty])$ by (1); so
  $F \leq P$, by relying on the order isomorphism
  $\Open^* \C \cong \Lhom \C$ with $\C \eqdef \Lform X$.  We have
  $F^{-1} (]1, \infty]) = \shadow {\mathcal U} {\one}$ and
  ${(\shd \one \alpha)}^{-1} (\shadow {\mathcal U} {\one}) \subseteq
  \shadow {\mathcal U} {\one}$ for every $\alpha \in {]0, 1]}$ by
  Lemma~\ref{lemma:P:shd} (4), so by (1) of the current lemma, $F$ is
  subnormalized.  Finally,
  $\mathcal U = {(\shd \one 1)}^{-1} (\mathcal U) \subseteq \shadow
  {\mathcal U} {\one} = F^{-1} (]1, \infty])$.
\end{proof}

\begin{corollary}
  \label{corl:P:subnorm}
  Given any subnormalized prevision $P$ on a topological space $X$ and
  a superlinear prevision $F_0 \leq P$, there is a subnormalized
  superlinear prevision $F$ such that $F_0 \leq F \leq P$.
\end{corollary}
\begin{proof}
  We apply Lemma~\ref{lemma:P:subnorm} (2) to
  $\mathcal U \eqdef F_0^{-1} (]1, \infty])$; $\mathcal U$ is open,
  and convex.  We obtain a subnormalized superlinear prevision
  $F \leq P$ such that $\mathcal U \subseteq F^{-1} (]1, \infty])$,
  whence $F_0 \leq F$ by the order isomorphism
  $\Open^* \C \cong \Lhom \C$ with $\C \eqdef \Lform X$.
\end{proof}

\begin{lem}
  \label{lemma:P:supnorm}
  Let $X$ be a topological space, $P \in \Pred X$, and $\mathcal D
  \eqdef P^{-1} ([0, 1])$.  The following two statements are
  equivalent:
  \begin{enumerate}
  \item for every $h \in \Lform X$, $P (\one+h) \geq 1+P(h)$;
  \item $P (\one) \geq 1$ and for every $\alpha \in {]0, 1]}$,
    ${(\shd \one \alpha)}^{-1} (\mathcal D) \subseteq \mathcal D$.
  \end{enumerate}
\end{lem}
\begin{proof}
  $(1) \limp (2)$.  By taking $h \eqdef 0$ in (1), $P (\one) \geq 1$.
  For every $\alpha \in {]0, 1]}$, for every
  $h \in {(\shd \one \alpha)}^{-1} (\mathcal D)$, we have
  $P (\alpha \cdot h + (1-\alpha) \cdot \one) \leq 1$.  By (1) and
  positive homogeneity,
  $P (\alpha \cdot h + (1-\alpha) \cdot \one) \geq \alpha P(h) +
  1-\alpha$.  Therefore $\alpha P (h) \leq \alpha$, whence
  $P (h) \leq 1$ and therefore $h \in \mathcal D$.

  $(2) \limp (1)$.  Let $r \in \Rp$ be such that $P (\one+h) < r$.  We
  show that $1+P(h) \leq r$, and taking infima over $r$ will prove
  (1).  We have $P (\one+h) \geq P (\one) \geq 1$, by monotonicity and
  (2).  In particular, $r > 1$, so
  $\alpha \eqdef \frac {r-1} r \in {]0, 1]}$.  Since $P (\one+h) < r$,
  $\frac 1 r \cdot h + \frac 1 r \cdot \one \in \mathcal D$.  But
  $\shd \one \alpha (\frac 1 {r-1} \cdot h) = \frac 1 r \cdot h +
  \frac 1 r \cdot \one$, so $\frac 1 {r-1} \cdot h$ is in
  ${(\shd \one \alpha)}^{-1} (\mathcal D)$, hence is in $\mathcal D$
  by (2).  This means that $P (\frac 1 {r-1} \cdot h) \leq 1$, namely
  that $1+P (h) \leq r$.
\end{proof}

\begin{lem}
  \label{lemma:P:norm}
  Given any normalized prevision $P$ on a topological space $X$ and a
  subnormalized sublinear prevision $F_0 \geq P$ such that
  $F_0 (\one) \geq 1$, there is a normalized sublinear prevision $F$ such
  that $F _0\geq F \geq P$.
\end{lem}
\begin{proof}
  Let $\mathcal D \eqdef P^{-1} ([0, 1])$.  Since $P$ is lower
  semicontinuous, $\mathcal D$ is a closed set.  Let
  $\mathcal C_0 \eqdef F_0^{-1} ([0, 1])$.  Using the order
  isomorphism $\Open^* \C \cong \Lhom \C$ with $\C \eqdef \Lform X$,
  $\mathcal C_0$ is a non-empty convex closed subset of $\Lform X$,
  and since $F_0 \geq P$, $\mathcal C_0 \subseteq \mathcal D$.  Let
  $\mathcal F$ be the family of convex subsets of $\Lform X$ included
  in $\mathcal D$, ordered by inclusion.  Since directed unions of
  convex sets are convex, $\mathcal F$ is a dcpo, hence an inductive
  poset.  By Zorn's Lemma, $\mathcal F$ contains a maximal element
  $\mathcal C \supseteq C_0$.  Since $\mathcal D$ is closed,
  $cl (\mathcal C)$ is also included in $\mathcal D$, and is convex,
  so it is in $\mathcal F$.  Since $\mathcal C$ is maximal in
  $\mathcal F$, it follows that $\mathcal C = cl (\mathcal C)$,
  therefore $\mathcal C$ is closed.  Since $\mathcal C_0$ is
  non-empty, so is $\mathcal C$.  Hence $\mathcal C$ is non-empty,
  convex and closed.  Using the order isomorphism
  $\Open^* \C \cong \Lhom \C$,
  the upper Minkowski functional $F \eqdef M^{\C \diff \mathcal C}$ is
  a sublinear prevision.  We also have
  $\mathcal C_0 \subseteq \mathcal C \subseteq \mathcal D$, whence
  $F_0 \geq F \geq P$.

  The set $\shadow {\mathcal C} {\one}$ is convex by
  Lemma~\ref{lemma:P:shd}, item~2.  We claim that it is included in
  $\mathcal D$.  For that, we need to show that
  ${(\shd \one \alpha)}^{-1} (\mathcal C) \subseteq \mathcal D$ for
  every $\alpha \in {]0, 1]}$.  But
  ${(\shd \one \alpha)}^{-1} (\mathcal C) \subseteq {(\shd \one
    \alpha)}^{-1} (\mathcal D)$ (since
  $\mathcal C \subseteq \mathcal D$) $\subseteq \mathcal D$, by using
  the $(1) \limp (2)$ implication of Lemma~\ref{lemma:P:supnorm},
  which applies since $P$ is normalized.

  Hence $\shadow {\mathcal C} {\one} \in \mathcal F$.
  $\shadow {\mathcal C} {\one}$ contains
  ${(\shd \one 1)}^{-1} (\mathcal C) = \mathcal C$, so
  $\mathcal C = \shadow {\mathcal C} {\one}$, since $\mathcal C$ is
  maximal in $\mathcal F$.  We have $\mathcal C = F^{-1} ([0, 1])$ by
  definition of Minkowski functionals,
  and $F (\one) \geq F_0 (\one) \geq 1$ so the $(2) \limp (1)$
  implication of Lemma~\ref{lemma:P:supnorm}, applied to $F$ instead
  of $P$, tells us that $F (\one+h) \geq 1+F (h)$ for every
  $h \in \Lform X$.  Since $F$ is subnormalized, it satisfies the
  reverse inequality, so $F$ is normalized.
\end{proof}

\section{Double hyperspaces and orthogonality relations}
\label{sec:double-hypersp-orth}

It will be profitable to revisit the de Brecht-Kawai construction of
the double hyperspace as an instance of Birkhoff's theory of
polarities, a practical way of forming Galois connections between
powersets \cite{Birkhoff}.  A \emph{Galois connection}
$E \colon \_^\perp \dashv {^\perp\_} \colon F$ between posets $E$ and
$F$ is a pair of antitonic maps $\_^\perp \colon E \to F$ and
${^\perp\_} \colon F \to E$ such that for all $A \in E$ and $B \in F$,
$A \subseteq {^\perp B}$ if and only if $B \subseteq A^\perp$.
Then $\_^\perp$ and $^\perp\_$ restrict to an order isomorphism
between $\Img {{^\perp \_}}$ and $\Img {\_^\perp}$, one being ordered
by inclusion and the other one by reverse inclusion.  Additionally,
$\Img {{^\perp \_}}$ is exactly the collection of elements $A$ such
that $A = {^\perp {(A^\perp)}}$, and $\Img {\_^\perp}$ is the
collection of elements $B$ such that $B = {(^\perp B)}^\perp$.

A polarity, or an orthogonality relation, is any binary relation
${\perp} \subseteq \QE \times \CF$ between elements of two sets $\QE$
and $\CF$.  We write $x \perp y$ for $(x, y) \in {\perp}$.  For every
subset $A$ of $\QE$, we let
$A^\perp \eqdef \{y \in \CF \mid \forall x \in A, x \perp y\}$, and
for every subset $B$ of $\CF$, we let
$^\perp B \eqdef \{x \in \QE \mid \forall y \in B, x \perp y\}$.  Then
$\_^\perp$ and $^\perp \_$ form a Galois connection between $\pow \QE$
and $\pow \CF$.  Notably, $A \subseteq {^\perp B}$ and
$B \subseteq A^\perp$ are both equivalent to: for every $x \in A$, for
every $y \in B$, $x \perp y$.

The de Brecht-Kawai construction relating $\Hoarez {\SVz X}$ and
$\Smythz {\HVz X}$ is an instance of this framework, as we show in
Lemma~\ref{lemma:perp:QH=HQ} below.  We need a technical lemma first;
a similar lemma, restricted to $\Hoare X$ and finite intersections,
appears as Lemma~5.7 in \cite{GLK-mscs10}.
We extend the notation $\Diamond U$ to all subsets $U$ of a space $X$,
not just the open sets, with the same definition
$\{C \in \Hoarez X \mid C \cap U \neq \emptyset\}$.
\begin{lem}
  \label{lemma:H:diaQ}
  Let $X$ be a topological space.  For every family
  ${(Q_i)}_{i \in I}$ of compact saturated subsets,
  $\bigcap_{i \in I} \Diamond {Q_i}$ is a compact saturated subset of
  $\HVz X$.
\end{lem}
\begin{proof}
  We first deal with the case of $\HVz X$.  The set
  $\bigcap_{i=1}^n \Diamond {Q_i}$ is clearly saturated.  We use
  Alexander's subbase lemma in order to show compactness.  Let
  ${(\Diamond U_j)}_{j \in J}$ be a cover of
  $\bigcap_{i \in I} \Diamond {Q_i}$ by subbasic open sets.  Let $C_0$
  be the complement of $\bigcup_{j \in J} U_j$.  $C_0$ is a closed set
  which is not in $\bigcup_{j \in J} \Diamond {U_j}$, so $C_0$ is not
  in $\bigcap_{i \in I} \Diamond {Q_i}$ either.  Hence $C_0$ does not
  intersect $Q_i$ for some $i \in I$.  This $Q_i$ must then be
  included in the complement of $C_0$, which is
  $\bigcup_{j \in J} U_j$.  Since $Q_i$ is compact, there is a finite
  subset $J'$ of $J$ such that $Q_i \subseteq \bigcup_{j \in J'} U_j$.
  Every element $C$ of $\bigcap_{i \in I} \Diamond {Q_i}$ intersects
  $Q_i$, hence $U_j$ for some $j \in J'$.  Therefore
  $\bigcap_{i \in I} \Diamond {Q_i} \subseteq \bigcup_{j \in J'}
  \Diamond {U_j}$.
%
\end{proof}

\begin{lem}
  \label{lemma:perp:QH=HQ}
  Let $X$ be a topological space, and let $\QE \eqdef \Smythz X$,
  $\CF \eqdef \HVz X$, and $Q \perp C$ if and only if
  $Q \cap C \neq \emptyset$.  Then
  $\Img {^\perp \_} \subseteq \Hoarez {\SVz X}$ and
  $\Img {\_^\perp} \subseteq \Smythz {\HVz X}$.  If $X$ is consonant,
  then $\Img {^\perp \_} = \Hoarez {\SVz X}$ and
  $\Img {\_^\perp} = \Smythz {\HVz X}$.
\end{lem}
\begin{proof}
  We first show that for every subset $B$ of $\CF$, $^\perp B$ is
  closed in $\SVz X$.  If $B$ is empty, then $^\perp B = \QE$, which
  is closed in $\SVz X$.  Let $B$ be non-empty.  Then
  ${^\perp B} = \bigcap_{C \in B} \{Q \in \QE \mid Q \cap C \neq
  \emptyset\}$.  This is the complement of
  $\bigcup_{C \in B} \{Q \in \QE \mid Q \cap C = \emptyset\} =
  \bigcup_{C \in B} \Box {(X \diff C)}$, hence it is open.  Therefore
  $\Img {^\perp \_} \subseteq \Hoarez {\SVz X}$.

  Second, we show that for every subset $A$ of $\QE$, $A^\perp$ is
  compact saturated in $\CF$.  We have
  $A^\perp = \bigcap_{Q \in A} \{C \in \CF \mid Q \cap C \neq
  \emptyset\} = \bigcap_{Q \in A} \Diamond Q$, and this is compact
  saturated by Lemma~\ref{lemma:H:diaQ}.  Hence
  $\Img {\_^\perp} \subseteq \Smythz {\HVz X}$.

  The restriction of ${^\perp \_}$ to
  $\SVz {\HVz X} \subseteq \pow {\CF}$ is the map $\tau_X$ of
  (\ref{eq:tau}), and the restriction of $\_^\perp$ to $\HVz {\SVz X}$
  is the map $\sigma_X$ of (\ref{eq:sigma}).  When $X$ is consonant,
  $\sigma_X$ is bijective, with inverse $\tau_X$, so we can write any
  element $\mathcal C$ of $\Hoarez {\SVz X}$ as
  $\tau_X (\sigma_X (\mathcal C)) = {({^\perp \mathcal C})}^\perp$,
  hence $\mathcal C \in \Img {\_^\perp}$.  It follows that
  $\Img {^\perp \_} = \Hoarez {\SVz X}$.  We show that
  $\Img {\_^\perp} = \Smythz {\HVz X}$ in a similar fashion.
\end{proof}

We now consider the case where $\QE = \SVc {\Pred^\rast_\lin X}$ and
$\CF = \HVc {\Pred^\rast_\lin X}$, with $Q \perp C$ defined by
$Q \cap C \neq \emptyset$.  By the representation theorem for
superlinear previsions, we can always write $Q$ as
$s^\rast_\super (F^-)$ for some unique $F^- \in \Pred^\rast_\super X$.
When $X$ is $\AN_\rast$-friendly, we can write $C$ as
$s^\rast_\sub (F^+)$ for some unique $F^+ \in \Pred^\rast_\sub X$, by
the representation theorem for sublinear previsions.
\begin{lem}
  \label{lemma:s:perp}
  Let $\rast$ be nothing, ``$\leq 1$'', or ``$1$''.  For every
  topological space $X$, for every $F^- \in \Pred^\rast_\super X$, for
  every $F^+ \in \Pred^\rast_\sub X$,
  $s^\rast_\super (F^-) \cap s^\rast_\sub (F^+) \neq \emptyset$ if and
  only if $F^- \leq F^+$.
\end{lem}
\begin{proof}
  $s^\rast_\super (F^-) \cap s^\rast_\sub (F^+) \neq \emptyset$ if and
  only if there is a $\Lambda \in \Pred^\rast_\lin X$ such that
  $F^- \leq \Lambda \leq F^+$.  This clearly implies $F^- \leq F^+$.
  Conversely, if $F^- \leq F^+$, there is a linear prevision $\Lambda$
  on $X$ such that $F^- \leq \Lambda \leq F^+$ by Keimel's sandwich
  theorem.  When $\rast$ is ``$\leq 1$'', we observe that $\Lambda$ is
  subnormalized: we have $\Lambda (\one) \leq F^+ (\one) \leq 1$,
  since $F^+$ is itself subnormalized, and then for every
  $h \in \Lform X$,
  $\Lambda (\one+h) = \Lambda (\one) + \Lambda (h) \leq 1 + \Lambda
  (h)$.  When $\rast$ is ``$1$'', $\Lambda$ is normalized: since both
  $F^-$ and $F^+$ are normalized,
  $1 = F^- (\one) \leq \Lambda (\one) \leq F^+ (\one)=1$, and
  therefore
  $\Lambda (\one+h) = \Lambda (\one) + \Lambda (h) = 1 + \Lambda (h)$.
\end{proof}

Hence we will instead focus on the case
$\QE \eqdef \Pred^\rast_\super X$, $\CF \eqdef \Pred^\rast_\sub X$,
defining $F^- \perp F^+$ by $F^- \leq F^+$.  This is equivalent to the
case mentioned before Lemma~\ref{lemma:s:perp} when $X$ is
$\AN_\rast$-friendly.  With this new definition, $\Img {\_^\perp}$
will be a collection of non-empty convex compact saturated subsets of
$\CF$ that we will write as $\QPred^\rast_\sub X$, which we will
explore in Section~\ref{sec:PX=QPAPX}, and $\Img {^\perp \_}$ will be
a collection of non-empty convex closed subsets of $\QE$ that we will
write as $\CPred^\rast_\super X$, and which we will explore in
Section~\ref{sec:PX=CPDPX}.  The fact that those are exactly
$\Img {\_^\perp}$ and $\Img {^\perp \_}$ will be obtained at the very
end, under a few assumptions, in Theorem~\ref{thm:HppQ}.

\section{The isomorphism $\Pred^\rast X \cong \QPred^\rast_\sub X$}
\label{sec:PX=QPAPX}

For the purposes of disambiguation, we will write $[h > r]$ or
$[h > r]^\rast$ for the set $\{P \in \Pred^\rast X \mid P (h) > r\}$,
where $h \in \Lform X$ and $r \in \Rp$, and we will write
$[h > r]^\rast_\sub$ for the corresponding subbasic open subset of
$\Pred^\rast_\sub X$, namely $[h > r] \cap \Pred^\rast_\sub X$, and
similarly with $[h > r]^\rast_\super$ or $[h > r]^\rast_\lin$.  The
subsets $[h > 1]^\rast$, $[h > 1]^\rast_\super$, $[h > 1]^\rast_\sub$
(that is, taking $r \eqdef 1$) already form subbases of the
appropriate spaces of previsions, because
$[h > r]^\rast = [(1/r) \cdot h > 1]^\rast$ when $r > 0$, and
$[h > 0]^\rast = \dcup_{r > 0} [h > r]^\rast$.

\begin{lemdef}
  \label{lemdef:minP}
  Let $\rast$ be nothing, ``$\leq 1$'' or ``$1$''.  For every space
  $X$, there is a continuous map
  $\minP_X \colon \SV {\Pred^\rast_\sub X} \to \Pred^\rast X$ defined by
  $\minP_X (\mathcal Q) (h) \eqdef \min_{F \in \mathcal Q} F (h)$ for
  every $\mathcal Q \in \SV {\Pred^\rast_\sub X}$ and every
  $h \in \Lform X$.  For every subbasic open set $[h > r]^\rast$ of
  $\Pred^\rast X$,
  $\minP_X^{-1} ([h > r]^\rast) = \Box {[h > r]^\rast_\sub}$.
\end{lemdef}
\begin{proof}
  The function $F \mapsto F (h)$ is lower semicontinuous, since the
  inverse image of $]r, \infty]$ is $[h > r]^\rast_\sub$.  Hence, for
  every compact saturated subset $\mathcal Q$ of $\Pred^\rast_\sub X$,
  the infimum $\inf_{F \in \mathcal Q} F (h)$ is reached, justifying
  the notation $\min_{F \in \mathcal Q} F (h)$.
  
  For every $\mathcal Q \in \SV {\Pred^\rast_\sub X}$,
  $\minP_X (\mathcal Q)$ is positively homogeneous.  In order to see
  that it is Scott-continuous, let ${(h_i)}_{i \in I}$ be a directed
  family with (pointwise) supremum $h$ in $\Lform X$. It is clear that
  $\minP_X (\mathcal Q)$ is monotonic, so
  $\dsup_{i \in I} \minP_X (\mathcal Q) (h_i) \leq \minP_X (\mathcal
  Q) (h)$.  Conversely, we claim that for every $r \in \Rp$ such that
  $r < \minP_X (\mathcal Q) (h)$, there is an $i \in I$ such that
  $r < \minP_X (\mathcal Q) (h_i)$.  The assumption
  $r < \minP_X (\mathcal Q) (h)$ rewrites as
  $\mathcal Q \subseteq [h > r]^\rast_\sub$, and it is easy to see
  that $[h > r]^\rast_\sub = \dcup_{i \in I} [h_i > r]^\rast_\sub$,
  using the Scott-continuity of previsions.  Since $\mathcal Q$ is
  compact (and using directedness),
  $\mathcal Q \subseteq [h_i > r]^\rast_\sub$ for some $i \in I$,
  namely $r < \minP_X (\mathcal Q) (h_i)$.

  Therefore $\minP_X (\mathcal Q)$ is a prevision.  If $\rast$ is
  ``$\leq 1"$, then for every $h \in \Lform X$,
  $\minP_X (\mathcal Q) (\one+h) = \min_{F \in \mathcal Q} F (\one+h)
  \leq \min_{F \in \mathcal Q} (1 + F (h)) = 1 + \minP_X (\mathcal Q)
  (h)$, so $\minP_X (\mathcal Q)$ is subnormalized; similarly, if
  $\rast$ is ``$1$'', $\minP_X (\mathcal Q)$ is normalized.

  For all $\mathcal Q \in \SV {\Pred^\rast_\sub X}$, $h \in \Lform X$,
  and $r \in \Rp$, $\minP_X (\mathcal Q) (h) > r$ if and only if for
  every $F \in \mathcal Q$, $F (h) > r$, and that is equivalent to
  $\mathcal Q \in \Box {[h > r]^\rast_\sub}$.
\end{proof}
The map $\minP_X$ is defined just like $r_\super$, but has a different
domain and codomain.

We will restrict $\minP_X$ to the following subspace
$\QPred^\rast_\sub X$.  This will turn out to be $\Img {\_^\perp}$,
namely the set of elements $\mathcal Q$ such that
$\mathcal Q = {^\perp {({\mathcal Q}^\perp)}}$, justifying the double
$\perp$ superscript.

\begin{definition}[$\QPred^\rast_\sub X$]
  \label{defn:QBox}
  Let $\rast$ be nothing, ``$\leq 1$'' or ``$1$''.  For every
  topological space $X$, for every $h \in \Lform X$, for every
  $r \in \Rp$, let
  $[h \geq r]^\rast_\sub \eqdef \{F \in \Pred^\rast_\sub X \mid F (h) \geq
  r\}$, and let $\QPred^\rast_\sub X$ be the subspace of
  $\SV {\Pred^\rast_\sub X}$ consisting of those intersections of the
  form $\bigcap_{i \in I} [h_i \geq 1]^\rast_\sub$ that are compact
  saturated and non-empty, where $I$ is an arbitrary set and where
  $h_i \in \Lform X$.
\end{definition}

\begin{lem}
  \label{lemma:QBox:Qcvx}
  Let $\rast$ be nothing, ``$\leq 1$'' or ``$1$''.  Let $X$ be a
  topological space.  Then every element $\QPred^\rast_\sub X$ is
  non-empty, convex and compact saturated.  $\QPred_\sub X$ is a
  subspace of $\SVc {\Pred_\sub X}$.
\end{lem}
\begin{proof}
  Let $\mathcal Q \eqdef \bigcap_{i \in I} [h_i \geq 1]^\rast_\sub$ be
  in $\QPred^\rast_\sub X$.  $\mathcal Q$ is compact saturated and
  non-empty by definition.  Each set $[h_i \geq 1]^\rast_\sub$ is
  convex: for all $F, F' \in [h_i \geq 1]^\rast_\sub$ and
  $a \in [0, 1]$, $a F (h_i) + (1-a) F' (h_i) \geq a+(1-a)=1$.  Hence
  $\mathcal Q$ is convex, being an intersection of convex sets.
  Finally, the topology of $\QPred_\sub X$ is the subspace topology
  arising from $\SV {\Pred_\sub X}$, hence also from
  $\SVc {\Pred_\sub X}$.
\end{proof}

\begin{lemdef}
  \label{lemdef:nimP}
  Let $\rast$ be nothing, ``$\leq 1$'' or ``$1$''.  Let $X$ be a
  topological space.  There is a map
  $\nimP^\rast_X \colon \Pred^\rast X \to \SV {\Pred^\rast_\sub X}$
  defined by
  $\nimP^\rast_X (P) \eqdef \{F \in \Pred^\rast_\sub X \mid F \geq P\}$,
  and:
  \begin{enumerate}
  \item $\nimP^\rast_X$ takes its values in $\QPred^\rast_\sub X$;
  \item for every $\mathcal Q \in \SV {\Pred^\rast_\sub X}$,
    $\mathcal Q \subseteq \nimP^\rast_X (\minP_X (\mathcal Q))$, with
    equality if and only if $\mathcal Q \in \QPred^\rast_\sub X$;
  \item a set $\mathcal Q$ is in $\QPred^\rast_\sub X$ if and only if it
    is of the form $\nimP^\rast_X (P)$ for some $P \in \Pred^\rast X$.
  \end{enumerate}
\end{lemdef}
\begin{proof}
  (1) For every $P \in \Pred^\rast X$,
  $\nimP^\rast_X (P) = \bigcap_{h \in \Lform X, r \in \Rp \diff \{0\},
    r< P (h)} [(1/r) h \geq 1]^\rast_\sub$, so $\nimP^\rast_X (P)$ is
  in $\QPred^\rast_\sub X$, provided that it is compact saturated and
  non-empty, which we now check.  We start with non-emptiness.

  When $\rast$ is nothing, let $F$ be the map
  $h \in \Lform X \mapsto \infty.\sup_{x \in X} h (x)$; in other words,
  $F (0)=0$ and $F (h) = \infty$ for every $h \neq 0$.  $F$ is
  Scott-continuous and positively homogeneous.  For all
  $h, h' \in \Lform X$, $F (h+h') = 0$ if and only if $h+h'=0$ if and
  only if $h=h'=0$ if and only if $F (h)=F(h')=0$, if and only if
  $F (h) + F (h')=0$; hence $F$ is a linear prevision.  It is clear
  that $F \geq P$, so $\nimP_X (P)$ is non-empty.

  When $\rast$ is ``$\leq 1$'' or ``$1$'', we define $F (h)$ as
  $\sup_{x \in X} h (x)$ instead.  This $F$ is a sublinear prevision,
  notably because of the inequality
  $\sup_{x \in X} (h (x) + h' (x)) \leq \sup_{x \in X} h (x) + \sup_{x
    \in X} h' (x)$.  For every $h \in \Lform X$,
  $F (\one + h) \leq 1+F (h)$, with equality if $X$ is non-empty.
  But, if $\rast$ is ``$1$'', $P$ is normalized, and therefore $X$ is
  non-empty, since the only prevision on the empty space is the zero
  prevision, which is not normalized.  Hence, whether $\rast$ is
  ``$\leq 1$'' or ``$1$'', $F$ is in $\Pred^\rast_\sub X$.  Finally,
  for every $h \in \Lform X$,
  $P (h) \leq P (\sup_{x \in X} h (x) \cdot \one) = \sup_{x \in X} h
  (x) . P (\one) \leq F (h)$, so $F \geq P$.  Hence
  $\nimP^\rast_X (P)$ is non-empty.

  The set $\nimP^\rast_X (P)$ is clearly upwards-closed, namely
  saturated.  Let us imagine that $\nimP^\rast_X (P)$ is not compact.
  By Alexander's subbase lemma, there is a cover of
  $\nimP^\rast_X (P)$ by subbasic open sets $[h_i > 1]^\rast_\sub$,
  $i \in I$ (where each $h_i$ is in $\Lform X$) with no finite
  subcover.  In particular, $I$ is infinite, hence non-empty.  For
  every non-empty finite subset $J$ of $I$, let $F_J$ be an element of
  $\nimP^\rast_X (P)$ that is not in
  $\bigcup_{i \in J} [h_i > 1]^\rast_\sub$.  Let
  $\mathcal C_J \eqdef F_J^{-1} ([0, 1])$: this is a closed subset of
  $\Lform X$, which is convex because $F_J$ is sublinear.  Since
  $F_J \in \nimP^\rast_X (P)$, namely $F_J \geq P$, $\mathcal C_J$ is
  included in $P^{-1} ([0, 1])$.  Also, for every $i \in J$,
  $F_J (h_i) \leq 1$, so $h_i \in \mathcal C_J$; and if $\rast$ is
  ``$\leq 1$'' or ``$1$'', then $F_J (\one) \leq 1$, so $\one$ is also
  in $\mathcal C_J$.  Since $\mathcal C_J$ is closed, convex and
  contains each $h_i$, $i \in J$ (resp., and $\one$ in case $\rast$ is
  ``$\leq 1$'' or ``$1$''), it also contains its closed convex hull
  $\mathcal C'_J \eqdef \clconv {\{h_i \mid i \in I\}}$ (resp.,
  $\clconv {(\{h_i \mid i \in I\} \cup \{\one\})}$).  Just like
  $\mathcal C_J$, $\mathcal C'_J$ contains each $h_i$, $i \in J$
  (resp., and $\one$ if $\rast$ is ``$\leq 1$'' or ``$1$''), and is
  included (in $\mathcal C_J$ hence) in $P^{-1} ([0, 1])$.
  Additionally, for any two non-empty finite subsets $J$ and $J'$ of
  $I$, if $J \subseteq J'$ then
  $\mathcal C'_J \subseteq \mathcal C'_{J'}$.  In other words, the
  family of sets $\mathcal C'_J$ where $J$ ranges over the non-empty
  finite subsets of $I$ is directed, a property that the family of
  sets $\mathcal C_J$ may not have enjoyed.  Let $\mathcal C'$ be the
  closure of the union of the sets $\mathcal C'_J$.  Since directed
  unions of convex sets are convex, and since closures of convex sets
  are convex, $\mathcal C'$ is closed and convex.  Every
  $\mathcal C'_J$ is included in the closed set $P^{-1} ([0, 1])$, so
  $\mathcal C' \subseteq P^{-1} ([0, 1])$, and $\mathcal C'$ contains
  every $h_i$ with $i \in I$ (resp., and $\one$ if $\rast$ is
  ``$\leq 1$'' or ``$1$'').  Since $I \neq \emptyset$, $\mathcal C'$
  is non-empty.  Hence the upper Minkowski functional
  $F_0 \eqdef M^{\Lform X \diff \mathcal C'}$ is a sublinear prevision.
  
  For every $i \in I$, $h_i \in \mathcal C' = F_0^{-1} ([0, 1])$, so
  $F_0 (h_i) \leq 1$ for every $i \in I$.  This shows that
  $F_0 \not\in \bigcup_{i \in I} [h_i > 1]^\rast_\sub$.  Since
  $\mathcal C' \subseteq P^{-1} ([0, 1])$,
  $\Lform X \diff \mathcal C' \supseteq P^{-1} (]1, \infty])$, and
  using the order isomorphism $\Open^* X \cong \Lhom X$, $F_0 \geq P$.
  Hence, when $\rast$ is nothing, we have found an element
  $F_0 \in \Pred_\sub X$ above $P$, namely an element of $\nimP_X (P)$,
  which is not in $\bigcup_{i \in J} [h_i > 1]_\sub$.

  When $\rast$ is ``$\leq 1$'', we note that $\one \in \mathcal C'$,
  so $F_0 (\one) \leq 1$.  Since $F_0$ is sublinear, for every
  $h \in \Lform X$,
  $F_0 (\one + h) \leq F_0 (\one) + F_0 (h) \leq 1 + F_0 (h)$, so $F_0$ is
  subnormalized; hence $F_0$ in $\nimP^{\leq 1}_X (P)$ but not in
  $\bigcup_{i \in I} [h_i > 1]^{\leq 1}_\sub$.

  When $\rast$ is ``$1$'', additionally,
  $F_0 (\one) \geq P (\one) = 1$, since $P$ is normalized.  By
  Lemma~\ref{lemma:P:norm}, there is a normalized sublinear prevision
  $F$ such that $F_0 \geq F \geq P$.  Then $F$ is in $\nimP^1_X (P)$,
  but not in $\bigcup_{i \in I} [h_i > 1]^1_\sub$.

  In all cases, we obtain a contradiction to our assumption that
  $\nimP^\rast_X (P) \subseteq \bigcup_{i \in I} [h_i >
  1]^\rast_\sub$.  Therefore $\nimP^\rast_X (P)$ is compact.

  (2) Let $\mathcal Q \in \SV {\Pred^\rast_\sub X}$.  For all
  $F \in \mathcal Q$ and $h \in \Lform X$,
  $F (h) \geq \min_{F \in \mathcal Q} F (h) = \minP_X (\mathcal Q)
  (h)$, so $F \in \nimP^\rast_X (\minP_X (\mathcal Q))$.  This shows
  that $\mathcal Q \subseteq \nimP^\rast_X (\minP_X (\mathcal C))$.

  If equality holds, or in general if $\mathcal Q = \nimP^\rast_X (P)$
  for some $P \in \Pred^\rast X$, then
  $\mathcal Q \in \QPred^\rast_\sub X$ by~(1).

  Conversely, let $\mathcal Q$ be a non-empty compact saturated set of
  the form $\bigcap_{i \in I} [h_i \geq 1]^\rast_\sub$, where
  $h_i \in \Lform X$, and let $P \eqdef \minP_X (\mathcal Q)$.  For
  every $F \in \nimP^\rast_X (P)$, for every $i \in I$,
  $F (h_i) \geq P (h_i) = \min_{F \in \mathcal Q} F (h_i) \geq 1$, so
  $F \in \mathcal Q$.  Therefore
  $\nimP^\rast_X (P) \subseteq \mathcal Q$ and hence
  $\mathcal Q = \nimP^\rast_X (\minP_X (P))$.

  (3) For every $\mathcal Q \in \QPred^\rast_\sub X$,
  $\mathcal Q = \nimP^\rast_X (P)$ where
  $P \eqdef \minP_X (\mathcal Q)$, by~(2).  Conversely, every set
  $\nimP^\rast_X (P)$ where $P \in \Pred^\rast X$ is in
  $\QPred^\rast_\sub X$ by~(1).
\end{proof}

\begin{corollary}
  \label{corl:nimP}
  Let $\rast$ be nothing, ``$\leq 1$'' or ``$1$'' and $X$ be a
  topological space.  Let $\QE \eqdef \Pred^\rast_\super X$,
  $\CF \eqdef \Pred^\rast_\sub X$, and let us define $F^- \perp F^+$ by
  $F^- \leq F^+$.  Then
  $\Img {\_^\perp} \subseteq \QPred^\rast_\sub X$.
\end{corollary}
\begin{proof}
  Let $A \subseteq \Pred^\rast_\super X$.  Then
  $A^\perp = \{F^+ \in \Pred^\rast_\sub X \mid \forall F^- \in A, F^-
  \leq F^+\}$.  Let $P (h) \eqdef \sup_{F^- \in A} F^- (h)$ for every
  $h \in \Lform X$.  As a pointwise supremum of lower semicontinuous
  maps, $P$ is lower semicontinuous.  It is positively homogeneous
  since every $F^- \in A$ is, so $P$ is a prevision on $X$.  Now
  $A^\perp = \{F^+ \in \Pred^\rast_\sub X \mid P \leq F^+\} =
  \nimP^\rast_X (P)$, and we conclude by using Lemma and
  Definition~\ref{lemdef:nimP}~(1).
\end{proof}

The following lemma is instrumental in understanding the topology of
$\QPred^\rast_\sub X$.  For every $U \in \Open {\Pred^\rast_\sub X}$,
we abbreviate $\Box U \cap \QPred^\rast_\sub X$ as $\Box^\pp U$.
\begin{lem}
  \label{lemma:just:Box:base}
  Let $\rast$ be nothing, ``$\leq 1$'' or ``$1$'', and let $X$ be a
  topological space.  For every $n \geq 1$ and for all
  $h_1, \cdots, h_n \in \Lform X$,
  \[
    \Box^\pp {\bigcup_{i=1}^n [h_i > 1]^\rast_\sub} = \bigcup_{\vec a
      \in \Delta_n} \Box^\pp {\left[\sum_{i=1}^n a_i h_i >
        1\right]^\rast_\sub}.
  \]
\end{lem}
\begin{proof}
  For every
  $\mathcal Q \in \bigcup_{\vec a \in \Delta_n} \Box^\pp
  {\left[\sum_{i=1}^n a_i h_i > 1\right]^\rast_\sub}$, we have that
  for every $F \in \mathcal Q$, there is an $\vec a \in \Delta_n$ such
  that $F (\sum_{i=1}^n a_i h_i) > 1$.  Since $F$ is sublinear,
  $\sum_{i=1}^n a_i F (h_i) > 1$, hence $F (h_i) > 1$ for some
  $i \in \{1, \cdots, n\}$: if instead we had $F (h_i) \leq 1$ for
  every $i \in \{1, \cdots, n\}$, then $\sum_{i=1}^n a_i F (h_i)$
  would be smaller than or equal to $1$.  We have obtained that
  $F \in \bigcup_{i=1}^n [h_i > 1]^\rast_\sub$, for every
  $F \in \mathcal Q$, so
  $\mathcal Q \in \Box^\pp {\bigcup_{i=1}^n [h_i > 1]^\rast_\sub}$.

  Conversely, let us assume that $\mathcal Q$ is not in
  $\bigcup_{\vec a \in \Delta_n} \Box^\pp {\left[\sum_{i=1}^n a_i h_i
      > 1\right]^\rast_\sub}$.  We will show that $\mathcal Q$ cannot
  be $\Box^\pp {\bigcup_{i=1}^n [h_i > 1]^\rast_\sub}$.  For every
  $\vec a \in \Delta_n$, there is an element
  $F_{\vec a} \in \mathcal Q$ such that
  $F_{\vec a} (\sum_{i=1}^n a_i h_i) \leq 1$.  In particular, letting
  $P \eqdef \minP_X (\mathcal Q)$, $P (\sum_{i=1}^n a_i h_i) \leq 1$
  for every $\vec a \in \Delta_n$.  This shows that
  $\conv {\{h_1, \cdots, h_n\}}$ is included in $P^{-1} ([0, 1])$, and
  since $P^{-1} ([0, 1])$ is closed, that
  $\mathcal C \eqdef \clconv {\{h_1, \cdots, h_n\}}$ is included in
  $P^{-1} ([0, 1])$.  $\mathcal C$ is a closed convex set, which is
  non-empty since $n \geq 1$.  Hence the upper Minkowski functional
  $F \eqdef M^{\Lform X \diff \mathcal C}$ of its complement is a
  sublinear prevision.
  $\Lform X \diff \mathcal C = F^{-1} (]1, \infty])$ contains
  $P^{-1} (]1, \infty])$ since $\mathcal C \subseteq P^{-1} ([0, 1])$;
  whence $F \geq P$.  We now use that $\mathcal Q$ is in
  $\QPred^\rast_\sub X$, not just in $\SV {\Pred^\rast_\sub X}$: by
  Lemma and Definition~\ref{lemdef:nimP}~(2),
  $\mathcal Q = \nimP^\rast_X (P)$, so $F \geq P$ entails that
  $F \in \mathcal Q$.  For each $i \in \{1, \cdots, n\}$,
  $h_i \in \mathcal C = F^{-1} ([0, 1])$, so
  $F \not\in \bigcup_{i=1}^n [h_i > 1]^\rast_\sub$.  Hence
  $\mathcal Q \not\in \Box^\pp {\bigcup_{i=1}^n [h_i >
    1]^\rast_\sub}$.
\end{proof}

\begin{lem}
  \label{lemma:QBox:lc}
  Let $\rast$ be nothing, ``$\leq 1$'' or ``$1$''.  Let $X$ be a
  topological space.  The sets $\Box^\pp {[h > 1]^\rast_\sub}$ where
  $h \in \Lform X$ form a base of the topology on
  $\QPred^\rast_\sub X$.
\end{lem}
\begin{proof}
  A subbase of the topology on $\Pred^\rast_\sub X$ is given by the
  sets $[h > 1]^\rast_\sub$ with $h \in \Lform X$.  Any open subset
  can then be written as a union of finite intersections of those.  We
  reorganize that as a directed union of finite unions of finite
  intersections, then as a directed union of finite intersections of
  finite unions, by distributing unions over intersections.  The
  $\Box$ operator commutes with directed unions and with finite
  intersections, so a base of the topology on $\QPred^\rast_\sub X$ is
  given by the sets
  $\Box {\bigcup_{i=1}^n [h_i > 1]^\rast_\sub} \cap \QPred^\rast_\sub
  X = \Box^\pp {\bigcup_{i=1}^n [h_i > 1]^\rast_\sub}$.

  We conclude because such sets can be rewritten as unions of sets of
  the form $\Box^\pp {[h > 1]^\rast_\sub}$: if $n \geq 1$, by
  Lemma~\ref{lemma:just:Box:base} we can take the functions
  $h \eqdef \sum_{i=1}^n a_i h_i$; if $n = 0$,
  $\Box^\pp {\bigcup_{i=1}^n [h_i > 1]^\rast_\sub}$ is empty, hence
  the union of the empty family.
\end{proof}


\begin{lem}
  \label{lemma:nimP:minP}
  Let $\rast$ be nothing, ``$\leq 1$'' or ``$1$''.  Let $X$ be a
  topological space.  For every $P \in \Pred^\rast X$,
  $\minP_X (\nimP^\rast_X (P)) = P$.
\end{lem}
\begin{proof}
  For every $h \in \Lform X$,
  $\minP_X (\nimP^\rast_X (P)) (h) = \min_{F \in \Pred^\rast_\sub X, F
    \geq P} F (h) \geq P (h)$, so
  $\minP_X (\nimP^\rast_X (P)) \geq P$.  If the inequality were
  strict, there would be an $h \in \Lform X$ such that
  $P (h) < \min_{F \in \Pred^\rast_\sub X, F \geq P} F (h)$.  We pick
  $r$ so that
  $P (h) < r < \min_{F \in \Pred^\rast_\sub X, F \geq P} F (h)$.
  Replacing $h$ by $(1/r) \cdot h$ if necessary, we can take $r = 1$.
  Hence $P (h) < 1 < \min_{F \in \Pred^\rast_\sub X, F \geq P} F (h)$.

  Let us deal with the case where $\rast$ is nothing first.  The
  downward closure $\mathcal C \eqdef \dc h$ is convex, closed,
  non-empty, and included in $P^{-1} ([0, 1])$ since $P (h) \leq 1$.
  Let $F$ be the upper Minkowski functional
  $M^{\Lform X \diff \mathcal C}$ of its complement.  Then $F$ is a
  sublinear prevision, and the complement of $\mathcal C$ is equal to
  $F^{-1} (]1, \infty])$, namely $\mathcal C = F^{-1} ([0, 1])$.
  Since $\mathcal C \subseteq P^{-1} ([0, 1])$,
  $P^{-1} (]1, \infty]) \subseteq F^{-1} (]1, \infty])$, so
  $P \leq F$.  Since $h \in \mathcal C$, $F (h) \leq 1$, so
  $\min_{F \in \Pred_\sub X, F \geq P} F (h) \leq 1$, which is
  impossible.

  When $\rast$ is ``$\leq 1$'' or ``$1$'', instead, we let
  $\mathcal C \eqdef \clconv {\{h, \one\}}$.  This is also convex,
  closed and non-empty.  For every element
  $a \cdot h + (1-a) \cdot \one$ of $\conv {\{h, \one\}}$,
  $P (a \cdot h + (1-a) \cdot \one) \leq a P (h) + (1-a)$ (since $P$
  is subnormalized) $\leq 1$ (because $P (h) \leq 1$).  Since
  $P^{-1} ([0, 1])$ is closed, the closure of $\conv {\{h, \one\}}$,
  which is $\mathcal C$, is included in $P^{-1} ([0, 1])$.  As above,
  we define $F$ as $M^{\Lform X \diff \mathcal C}$, so $F$ is a
  sublinear prevision such that $P \leq F$ and $F (h) \leq 1$ (because
  $h \in \mathcal C$ and $F = M^{\Lform X \diff \mathcal C}$).
  Additionally, since $\one \in \mathcal C$, $F (\one) \leq 1$, so
  $F (\one + h) \leq F (\one) + F (h) \leq 1 + F (h)$ for every
  $h \in \Lform X$, in other words $F$ is subnormalized.  When $\rast$
  is ``$1$'', $P \leq F$ entails $1 = P (\one) \leq F (\one)$; by
  Lemma~\ref{lemma:P:norm}, we can replace $F$ by a smaller normalized
  prevision that is still above $P$.  Hence, whether $\rast$ is
  ``$\leq 1$'' or ``$1$'',
  $\min_{F \in \Pred^\rast_\sub X, F \geq P} F (h) \leq 1$, which is
  impossible.
\end{proof}

\begin{corollary}
  \label{corl:nimP:cont}
  Let $\rast$ be nothing, ``$\leq 1$'' or ``$1$''.  Let $X$ be a
  topological space.  The function $\nimP^\rast_X$ is continuous from
  $\Pred^\rast X$ to $\QPred^\rast_\sub X$, hence also to
  $\SV {\Pred^\rast_\sub X}$.
\end{corollary}
\begin{proof}
  Using Lemma~\ref{lemma:QBox:lc}, it suffices to show that
  ${(\nimP^\rast_X)}^{-1} (\Box^\pp {[h > 1]^\rast_\sub})$ is open in
  $\Pred^\rast X$ for every $h \in \Lform X$.  Since $\nimP^\rast_X$
  takes its values in $\QPred^\rast_\sub X$ by Lemma and
  Definition~\ref{lemdef:nimP}, this set is simply
  ${(\nimP^\rast_X)}^{-1} (\Box {[h > 1]^\rast_\sub})$, and that is
  equal to ${(\nimP^\rast_X)}^{-1} (\minP_X^{-1} ([h > r]^\rast))$ by
  Lemma and Definition~\ref{lemdef:minP}.  By
  Lemma~\ref{lemma:nimP:minP},
  $\minP_X \circ \nimP^\rast_X = \identity {\Pred^\rast X}$, so
  ${(\nimP^\rast_X)}^{-1} (\Box^\pp {[h > 1]^\rast_\sub})$ is equal to
  $[h > r]^\rast$.
\end{proof}

We investigate what operations are preserved by $\minP_X$, although we
could proceed directly to our main theorems~\ref{thm:min:nim:retr}
and~\ref{thm:min:nim:iso} at this point.  For every topological cone
$\C$, $\SVc \C$ is a topological cone with zero equal to 
$\upc 0 = \SVc \C$, addition defined by
$Q_1 + Q_2 \eqdef \upc {\{x+y \mid x \in Q_1, y \in Q_2\}}$, scalar
multiplication by $a \cdot Q \eqdef \upc \{a \cdot x \mid x \in Q\}$
(or: $a \cdot Q \eqdef \{a \cdot x \mid x \in Q\}$ if $a > 0$,
$\upc 0$ otherwise).  It is also an inf-semilattice with respect to
its specialization ordering $\supseteq$: the
inf 
$Q_1 \sqcap Q_2 \eqdef \upc \conv {(Q_1 \cup Q_2)}$, and $\sqcap$ is
continuous and distributes over $+$ and $\cdot$
\cite[Theorem~11.1]{Keimel:topcones2}.  This applies to
$\C \eqdef \Pred_\sub X$, by Proposition~\ref{prop:PX:cone}.
\begin{prop}
  \label{prop:Qcvx:cone}
  Let $\B$ be a convex subspace of a topological cone $\C$.  The
  function that maps every $Q \in \SVc \B$ to its upward closure
  $\upc Q$ in $\C$ is an affine topological embedding, allowing us to
  see $\SVc \B$ as a convex subspace of $\SVc \C$ up to affine
  homeomorphism.  If $0 \in \B$, then the zero $\upc 0$ of $\SVc \C$
  belongs to $\SVc \B$ up to that homeomorphism.
\end{prop}
\begin{proof}
  Let $i$ be the map $Q \mapsto \upc Q$; $i$ is injective since
  $i (Q_1) = i (Q_2)$ implies
  $Q_1 = i (Q_1) \cap \B = i (Q_2) \cap \B = Q_2$.  For every open
  subset $U$ of $\C$, $i^{-1} (\Box^\cvx U)$ is equal to
  $\{Q \in \SVc \B \mid Q \subseteq U\}$, which is equal
  $\Box^\cvx {(U \cap \B\}}$, a basic open subset of $\SVc \B$.
  Conversely, every basic open subset of $\SVc \B$ is of that form,
  hence equal to the inverse of some basic open subset $\Box^\cvx U$
  of $\C$.  This shows that $i$ is a homeomorphism onto its image.  We
  claim that $i$ is affine; it will follow that the image of $\SVc \B$
  is convex.  Let us write $\upc_\B$ for upward closure in $\B$ and
  $\upc_\C$ for upward closure in $\C$.  For all
  $Q_1, Q_2 \in \SVc \B$ and $a \in [0, 1]$,
  $i (Q_1) +_a i (Q_2) = \upc_\C (\upc_\C {Q_1} +_a \upc_\C {Q_2})$ is
  equal to $\upc_\C {Q_1}$ if $a=1$, to $\upc_\C {Q_2}$ if $a=0$, and
  to $\upc_\C {\{x +_a y \mid x \in Q_1, y \in Q_2\}}$ otherwise,
  because $+_a$ is continuous hence monotonic with respect to the
  specialization ordering.  We compare this to $i (Q_1 +_a Q_2)$.  The
  latter is also equal to $\upc_\C {Q_1}$ if $a=1$ and to
  $\upc_\C {Q_2}$ if $a=0$.  Otherwise, it is equal to
  $\upc_\C {\upc_\B {\{a \cdot x + (1-a) \cdot y \mid x \in Q_1, y \in
      Q_2\}}}$, namely to
  $\upc_\C {\{x +_a y \mid x \in Q_1, y \in Q_2\}}$ because the
  specialization ordering of the subspace $\B$ is the restriction of
  that of $\C$.

  Finally, if $0 \in \B$, then the zero $\upc 0$ is equal to the image
  of the zero of $\SVc \B$ (the upward closure of $0$ in $\B$) under
  $i$.
\end{proof}

In particular, $\SVc {\Pred^\rast_\sub X}$ has all non-empty finite
infima.
$\Pred^\rast X$ also has non-empty finite infima, which are computed
pointwise, see Proposition~\ref{prop:PX:sup}.
\begin{lem}
  \label{lemma:minP:inf}
  Let $\rast$ be nothing, ``$\leq 1$'' or ``$1$'', and $X$ be a
  topological space.  The map $\minP_X$ preserves non-empty finite
  infima from $\SVc {\Pred^\rast_\sub X}$ to $\Pred^\rast X$.
\end{lem}
\begin{proof}
  It suffices to consider binary infima.  Let
  $\mathcal Q_1, \mathcal Q_2 \in \SVc {\Pred^\rast_\sub X}$, and
  $\mathcal Q \eqdef \mathcal Q_1 \cup \mathcal Q_2$, so that
  $\mathcal Q_1 \sqcap \mathcal Q_2 = \upc \conv {\mathcal Q}$.  For
  every $h \in \Lform X$,
  $\minP_X (\upc \conv {\mathcal Q}) (h) = \min_{F \in \conv {\mathcal
      Q}} F (h)$ (since $F$ is monotonic).  This is less than or equal
  to $\min_{F \in \mathcal Q} F (h)$, but also larger than or equal to
  it: for every $F \in \conv {\mathcal Q}$, we write $F$ as
  $\sum_{i=1}^n a_i F_i$ where $\vec a \in \Delta_n$ and each $F_i$ is
  in $\mathcal Q$; then $\sum_{i=1}^n a_i F_i (h)$ must be larger than
  or equal to the least of the values $F_i (h)$,
  $i \in \{1, \cdots, n\}$.  Finally,
  $\min_{F \in \mathcal Q} F (h) = \min (\min_{F \in \mathcal Q_1} F
  (h), \allowbreak \min_{F \in \mathcal Q_2} F (h)) = \inf (\minP_X
  (\mathcal Q_1) (h), \minP_X (\mathcal Q_2) (h))$.
\end{proof}

We turn to preservation of convex combinations.  Given two
semitopological cones $\C$ and $\C'$, and convex subspaces $\B$ of
$\C$ and $\B'$ of $\C'$, we say that $f \colon \B \to \B'$ is affine
if it preserves convex combinations, extending the notion of affine
maps on cones.  When $\B$ and $\B'$ contain the $0$ element of their
respective cones, we call a map $f \colon \B \to \B'$ \emph{strict} if
and only if $f (0)=0$, and we call strict affine maps \emph{linear}.
If $f$ is linear in this sense, then $f (a \cdot x) = a f (x)$ for all
$a \in [0, 1]$ and $x \in \B$, since
$f (a \cdot x) = f (x +_a 0) = f (x) +_a f (0) = a f (x)$.  The notion
agrees with our previous notion of linearity on cones: if $\B=\C$,
then every strict affine map $f \colon \C \to \C'$ is linear in the
usual sense: $f (a \cdot x) = a f (x)$ for all $a \in \Rp$ and
$x \in \C$ (we have just seen this for $a \in [0, 1]$; if $a > 1$,
then
$f (a \cdot x) = a.(1/a) f (a \cdot x) = a f ((1/a) \cdot a \cdot x) =
a f (x)$) and $f (x+y) = f (x) + f (y)$ for all $x, y \in C$ (since
$f (x+y) = 2 f (x +_{1/2} y) = 2 (f (x) +_{1/2} f (y)) = f (x) +
f(y))$).
\begin{lem}
  \label{lemma:minP:+}
  Let $\rast$ be nothing, ``$\leq 1$'' or ``$1$'' and $X$ be a
  topological space.  The map $\minP_X$ is affine from
  $\SVc {\Pred^\rast_\sub X}$ to $\Pred^\rast X$, and linear if
  $\rast$ is nothing or ``$\leq 1$''.
\end{lem}
\begin{proof}
  For all $\mathcal Q_1, \mathcal Q_2 \in \SVc {\Pred^\rast_\sub X}$, for
  every $a \in [0, 1]$, for every $h \in \Lform X$,
  $\minP_X (\mathcal Q_1 +_a \mathcal Q_2) (h) = \minP_X (\{F \in
  \Pred^\rast_\sub X \mid \exists F_1 \in \mathcal Q_1, F_2 \in \mathcal
  Q_2, F \geq F_1 +_a F_2\}) (h) = \min_{F_1 \in \mathcal Q_1, F_2 \in
    \mathcal Q_2} (a\,F_1 (h) + (1-a)\, F_2 (h)) = a\, \min_{F_1 \in
    \mathcal Q_1} F_1 (h) + (1-a) \, \min_{F_2 \in \mathcal Q_2} F_2
  (h)$, and this is the value of
  $\minP_X (\mathcal Q_1) +_a \minP_X (\mathcal Q_2)$ applied to $h$.

  When $\rast$ is nothing or ``$\leq 1$'', $\minP_X$ maps the zero
  element $\Pred^\rast_\sub X$ of $\SVc {\Pred^\rast_\sub X}$ to the
  prevision $h \mapsto \min_{P \in \Pred^\rast_\sub X} P (h)$, namely
  the zero prevision, since we can take $P \eqdef 0$ in the min.
  Therefore $\minP_X$ is strict, and hence, linear.
\end{proof}

\begin{thm}
  \label{thm:min:nim:retr}
  Let $\rast$ be nothing, ``$\leq 1$'' or ``$1$''.  For every
  topological space $X$, $\Pred^\rast X$ is a retract of
  $\SVc {\Pred^\rast_\sub X}$ through $\minP_X$ and $\nimP^\rast_X$.  The
  retraction $\minP_X$ is affine (linear if $\rast$ is nothing or
  ``$\leq 1$'') and preserves non-empty finite infima.
\end{thm}
\begin{proof}
  The map $\minP_X$ is affine (resp.\ linear) by
  Lemma~\ref{lemma:minP:+}, non-empty-finite-inf-preserving by
  Lemma~\ref{lemma:minP:inf} and continuous by Lemma and
  Definition~\ref{lemdef:minP}; $\nimP^\rast_X$ is continuous by
  Corollary~\ref{corl:nimP:cont}, and a right inverse to $\minP_X$ by
  Lemma~\ref{lemma:nimP:minP}.
\end{proof}

\begin{thm}[First isomorphism]
  \label{thm:min:nim:iso}
  Let $\rast$ be nothing, ``$\leq 1$'' or ``$1$''.  For every
  topological space $X$, $\QPred^\rast_\sub X$ and $\Pred^\rast X$ are
  homeomorphic through $\minP_X$ and $\nimP^\rast_X$.
\end{thm}
\begin{proof}
  By Lemma and Definition~\ref{lemdef:minP}, $\minP_X$ defines (by
  restriction) a continuous map from $\QPred^\rast_\sub X$ to
  $\Pred^\rast X$; $\nimP^\rast_X$ is continuous by
  Corollary~\ref{corl:nimP:cont}.  They are inverse of each other by
  Lemma and Definition~\ref{lemdef:nimP}, item~2, and by
  Lemma~\ref{lemma:nimP:minP}.
\end{proof}

\begin{remark}
  \label{rem:first:double:DN}
  If $X$ is $\AN_\rast$-friendly, then
  $\Pred^\rast_{\sub} X \cong \HVc {\Pred^\rast_{\lin} X}$
  \cite[Theorem~4.11]{JGL-mscs16,JGL:iso:err}.  Then
  Theorem~\ref{thm:min:nim:iso} expresses $\Pred^\rast X$ as a kind of
  double hyperspace over $\Pred^\rast_\lin X$ ($\cong \Val_\rast X$),
  as we hinted in the introduction.
\end{remark}

\begin{remark}
  \label{rem:min:nim}
  Theorem~\ref{thm:min:nim:iso} (or already
  Lemma~\ref{lemma:nimP:minP}) implies that any prevision in
  $\Pred^\rast X$ can be written as a pointwise infimum of sublinear
  previsions in $\Pred^\rast_\sub X$.  This applies with no constraint
  on $X$.
\end{remark}

\begin{remark}
  \label{rem:QBox}
  We have seen that
  $\QPred^\rast_\sub X \subseteq \SVc {\Pred^\rast_\sub X}$.  The
  inclusion is strict in general: any of the following two examples
  will provide counterexamples.
\end{remark}

\begin{example}
  \label{ex:QBox:inf}
  A consequence of Theorem~\ref{thm:min:nim:iso} is that
  $\QPred^\rast_\sub X$ has all non-empty finite infima.  But they can
  be very different than in the enclosing space
  $\SVc {\Pred^\rast_\sub X}$.  
  Let $X \eqdef \{0, 1\}$ with the discrete topology,
  $\Lambda (h) \eqdef \frac 1 2 (h (0) + h (1))$,
  $P_1 (h) \eqdef \max (h (0), \Lambda (h))$,
  $P_2 (h) \eqdef \max (h (1), \Lambda (h))$ for every
  $h \in \Lform X$.  Then:
  \begin{enumerate}
  \item $P_1$, $P_2$ are normalized (sublinear) previsions on $X$ and
    $\Lambda$ is the infimum of $P_1$ and $P_2$ in $\Pred^\rast X$.
  \item We claim that there is no $a \in [0, 1]$ such that
    $\Lambda \geq a \cdot P_1 + (1-a) \cdot P_2$.  In order to see
    this, for every prevision $P$ on $X$, let
    $\underline P \colon [0, 1] \mapsto [0, 1]$ be defined by letting
    $\underline P (t)$ be $P$ applied to the function $h_t$ mapping
    $0$ to $t$ and $1$ to $1-t$.  If
    $\Lambda \geq a \cdot P_1 + (1-a) \cdot P_2$, then
    $\underline\Lambda \geq a \cdot \underline P_1 + (1-a) \cdot
    \underline P_2$.  For every $t \in [0, 1]$,
    $\underline \Lambda (t) = \frac 1 2$,
    $\underline P_1 (t) = \max (t, 1/2)$ and
    $\underline P_2 (t) = \max (1-t, 1/2)$.  When
    $t = \frac 1 2 + \epsilon$ with $\epsilon > 0$,
    $a \underline P_1 (t) + (1-a) \underline P_2 (t) = a (\frac 1 2 +
    \epsilon) + (1-a) \frac 1 2 = \frac 1 2 + a \epsilon$, which is
    smaller than or equal to $\underline \Lambda (t)$ only when $a=0$.
    When $t = \frac 1 2 - \epsilon$ with $\epsilon > 0$,
    $a \underline P_1 (t) + (1-a) \underline P_2 (t) = a \frac 1 2 +
    (1-a) (\frac 1 2 + \epsilon) = \frac 1 2 + (1-a) \epsilon$, which
    is smaller than or equal to $\underline \Lambda (t)$ only when
    $a=1$.  But we cannot have $a=0$ and $a=1$ at the same time.
  \item Hence
    $\nimP^\rast_X (\inf (P_1, P_2)) \neq \nimP^\rast_X (P_1) \sqcap
    \nimP^\rast_X (P_2)$, so $\nimP^\rast_X$ does not map binary
    infima in $\Pred^\rast X$ to binary infima in
    $\SVc {\Pred^\rast_\sub X}$.  Indeed, first, we recall that
    $\nimP^\rast_X (P_1) \sqcap \nimP^\rast_X (P_2) = \upc \conv
    {(\nimP^\rast_X (P_1) \cup \nimP^\rast_X (P_2))}$.  Second, since
    $\Lambda$ is linear hence sublinear (and normalized),
    $\Lambda \in \nimP^\rast_X (\inf (P_1, P_2))$.  If
    $\Lambda \in \upc \conv {(\nimP^\rast_X (P_1) \cup \nimP^\rast_X
      (P_2))}$, $\Lambda$ would be larger than or equal to
    $a \cdot F_1 + (1-a) \cdot F_2$ for some sublinear previsions
    $F_1 \geq P_1$ and $F_2 \geq P_2$ and some $a \in [0, 1]$.  Then
    $\Lambda \geq a \cdot P_1 + (1-a) \cdot P_2$, contradicting~(2).
  \item The infimum of $\nimP^\rast_X (P_1)$ and $\nimP^\rast_X (P_2)$
    in $\QPred^\rast_\sub X$ is
    $\nimP^\rast_X (\inf (P_1, \allowbreak P_2))$, and by (3) it
    differs from the infimum
    $\nimP^\rast_X (P_1) \sqcap \nimP^\rast_X (P_2)$ taken in
    $\SVc {\Pred^\rast_\sub X}$.  \qed
  \end{enumerate}
\end{example}

\begin{example}
  \label{ex:QBox:+}
  Another consequence of Theorem~\ref{thm:min:nim:iso} is that
  $\QPred_\sub X$ is a cone, by transporting the cone structure over
  from $\Pred X$.  But addition can be very different than in the
  enclosing cone $\SVc {\Pred_\sub X}$, and $\nimP_X$ is not in
  general linear from $\Pred X$ to $\SVc {\Pred_\sub X}$.  In order to
  see this, let $X \eqdef \{0, 1\}$ with the discrete topology,
  $\Lambda (h) \eqdef h (0)+h(1)$,
  $P_1 (h) \eqdef \max (h (0), h (1))$ and
  $P_2 (h) \eqdef \min (2 h (0) + h (1), h (0) + 2 h (1))$ for every
  $h \in \Lform X$.  Then:
  \begin{enumerate}
  \item $P_1$ is a sublinear prevision and $P_2$ is a superlinear
    prevision.
  \item For every sublinear prevision $F_2 \geq P_2$, there is a
    linear prevision $\Lambda_2$ such that
    $F_2 \geq \Lambda_2 \geq P_2$, by Keimel's sandwich theorem.
    Using the same notation as in Example~\ref{ex:QBox:inf},
    $\underline \Lambda_2 (t) = a t + b (1-t)$ for every
    $t \in [0, 1]$, where $a \eqdef \underline \Lambda_2 (1)$ and
    $b \eqdef \underline \Lambda_2 (0)$.  Then $a + b \geq 3$, since
    $\underline \Lambda_2 (1/2) \geq \underline P_2 (1/2)$.
  \item In the situation of (2), it is impossible that
    $\underline P_1 + \underline F_2 \leq 2+\epsilon$, for any
    $\epsilon \in [0, \frac 1 2[$.  Otherwise, $F_2 \geq \Lambda_2$
    would imply
    $\underline P_1 + \underline \Lambda_2 \leq 2+\epsilon$.  We have
    $\underline P_1 (t) = \max (t, 1-t)$ for every $t \in [0, 1]$.
    Therefore $\max (t, 1-t) + a t + b (1-t) \leq 2+\epsilon$ for
    every $t \in [0, 1]$.  For $t \eqdef 0$, this entails
    $1+b \leq 2+\epsilon$, hence $b \leq 1+\epsilon$.  For
    $t \eqdef 1$, this entails $1+a \leq 2+\epsilon$, hence
    $a \leq 1+\epsilon$.  By summing the two,
    $a+b \leq 2 + 2 \epsilon$.  But $a+b \geq 3$ and $\epsilon < 1/2$.
  \item $\Lambda$ is in $\nimP_X (P_1 +_{\frac 1 2} P_2)$ but not in
    $\nimP_X (P_1) +_{\frac 1 2} \nimP_X (P_2)$.  Indeed, first,
    $\Lambda = P_1 +_{\frac 1 2} P_2$, as one checks by distinguishing
    the cases $h (0) < h (1)$ and $h (0) \geq h (1)$.
%
    If $\Lambda$ were in $\nimP_X (P_1) +_{\frac 1 2} \nimP_X (P_2)$,
    then there would be two sublinear previsions $F_1 \geq P_1$ and
    $F_2 \geq P_2$ such that
    $\Lambda \geq \frac 1 2 F_1 + \frac 1 2 F_2$.  In particular,
    $\Lambda \geq \frac 1 2 P_1 + \frac 1 2 F_2$, so
    $\underline \Lambda (t) = 1 \geq \frac 1 2 \underline P_1 (t) +
    \frac 1 2 \underline F_2 (t)$ for every $t \in [0, 1]$.  This
    contradicts~(3) (with $\epsilon \eqdef 0$).
  \item Hence $\nimP_X$ is not affine from $\Pred X$ to
    $\SVc {\Pred_\sub X}$.  Additionally, convex combinations of
    elements of $\QPred_\sub X$, as computed in $\SVc {\Pred_\sub X}$,
    may fail to be in $\QPred_\sub X$.  In particular,
    $\QPred^\rast_\sub X$ is a proper subspace of
    $\SVc {\Pred^\rast_\sub X}$.  In order to see this, let
    $\mathcal Q_1 \eqdef \nimP_X (P_1)$ and
    $\mathcal Q_2 \eqdef \nimP_X (P_2)$.  Then
    $\mathcal Q_1 +_{\frac 1 2} \mathcal Q_2 \neq \nimP_X (P_1
    +_{\frac 1 2} P_2)$ by (4).  But, if
    $\mathcal Q_1 +_{\frac 1 2} \mathcal Q_2$ were in $\QPred_\sub X$,
    we would have
    $\mathcal Q_1 +_{\frac 1 2} \mathcal Q_2 = \nimP_X (\minP_X
    (\mathcal Q_1 +_{\frac 1 2} \mathcal Q_2))$ by
    Lemma~\ref{lemma:nimP:minP}, which is equal to
    $\nimP_X (\minP_X (\mathcal Q_1) +_{\frac 1 2} \minP_X (\mathcal
    Q_2))$ by Lemma~\ref{lemma:minP:+}, and therefore to
    $\nimP_X (P_1 +_{\frac 1 2} P_2)$ by Lemma~\ref{lemma:nimP:minP}.
    \qed
  \end{enumerate}
\end{example}


\section{The isomorphism $\Pred^\rast X \cong \CPred^\rast_\super X$}
\label{sec:PX=CPDPX}

Instead of considering convex compact saturated subsets of
$\Pred^\rast_\sub X$, we now consider convex closed subsets of
$\Pred^\rast_\super X$.  The treatment is similar, although some
additional assumptions will be needed.

\begin{lemdef}
  \label{lemdef:supP}
  Let $\rast$ be nothing, ``$\leq 1$'' or ``$1$''.  For every space
  $X$, there is a continuous map
  $\supP_X \colon \HV {\Pred^\rast_\super X} \to \Pred^\rast X$ defined by
  $\supP_X (\mathcal C) (h) \eqdef \sup_{F \in \mathcal C} F (h)$ for
  every $\mathcal C \in \HV {\Pred^\rast_\super X}$ and every
  $h \in \Lform X$.  For every subbasic open set $[h > r]^\rast$ of
  $\Pred^\rast X$,
  $\supP_X^{-1} ([h > r]^\rast) = \Diamond {[h > r]^\rast_\super}$.
\end{lemdef}
\begin{proof}
  Clearly, $\supP_X (\mathcal C)$ is a prevision for every
  $\mathcal C \in \Hoare {\Pred^\rast_\super X}$, which is
  subnormalized (resp.\ normalized) if $\rast$ is ``$\leq 1$'' (resp.\
  ``$1$'').  For every $\mathcal C \in \Hoare {\Pred^\rast_\super X}$,
  $\supP_X (\mathcal C) \in [h > r]^\rast$ if and only if
  $\sup_{F \in \mathcal C} F (h) > r$, if and only if $F (h) > r$ for
  some $F \in \mathcal C$, if and only if
  $\mathcal C \in \Diamond {[h > r]^\rast_\super}$: so $\supP_X$ is
  continuous.
\end{proof}
The map $\supP_X$ is defined just like $r_\sub$, but has a different
domain and codomain.

\begin{definition}[$\CPred^\rast_\super X$]
  \label{defn:HDia}
  Let $\rast$ be nothing, ``$\leq 1$'' or ``$1$''.  For every
  topological space $X$, for every $h \in \Lform X$, for every
  $r \in \Rp$, let
  $[h \leq r]^\rast_\super \eqdef \{F \in \Pred^\rast_\super X \mid F (h) \leq
  r\}$, and let $\CPred^\rast_\super X$ be the subspace of
  $\HV {\Pred^\rast_\super X}$ consisting of those intersections of the
  form $\bigcap_{i \in I} [h_i \leq 1]^\rast_\super$ that are non-empty,
  where $I$ is an arbitrary set and where $h_i \in \Lform X$.
\end{definition}

\begin{lem}
  \label{lemma:Hdia:Hcvx}
  Let $\rast$ be nothing, ``$\leq 1$'' or ``$1$'' and $X$ be a
  topological space.  Every element of $\CPred^\rast_\super X$ is
  non-empty, closed and convex.  $\CPred_\super X$ is a subspace of
  $\HVc {\Pred_\super X}$.
\end{lem}
\begin{proof}
  For all $h \in \Lform X$ and $r \in \Rp$, $[h \leq r]^\rast_\super$
  is the complement of
  $[h > r]^\rast_\super$.
  As such, $[h \leq r]^\rast_\super$ is closed and convex, so every
  intersection of such sets is closed and convex.  Since we only keep
  those that are non-empty in $\CPred^\rast_\super X$, they are all in
  $\HVc {\Pred^\rast_\super X}$.  Finally, the topology of
  $\CPred_\super X$ was defined as the subspace topology arising from
  $\HV {\Pred_\super X}$, hence also from $\HVc {\Pred_\super X}$.
\end{proof}


\begin{lem}
  \label{lemma:qusP:0}
  Let $\rast$ be nothing, ``$\leq 1$'' or ``$1$'' and $X$ be a
  topological space, which we assume compact and non-empty if $\rast$
  is ``$1$''.  Let $P$ be a function from $\Lform X$ to $\creal$,
  which we assume monotonic, positively homogenous and such that
  $P (\one)=1$ in the case that $\rast$ is ``$1$''.  Then
  $\{F \in \Pred^\rast_\super X \mid F \leq P\} \in \CPred^\rast_\super
  X$.
\end{lem}
\begin{proof}
  Let
  $\mathcal C \eqdef \{F \in \Pred^\rast_\super X \mid F \leq P\}$.
  $\mathcal C$ is equal to
  $\bigcap_{h \in \Lform X, P (h) < r < \infty} [(1/r) \cdot h \leq
  1]^\rast_\super$.  Indeed, for every $F \in \mathcal C$, for every
  $h \in \Lform X$, for every $r \in \Rp$ such that $P (h) < r$,
  $F (h) \leq P (h) < r$, so $F ((1/r) \cdot h) < 1$.  Conversely,
  if $F ((1/r) \cdot h) \leq 1$ for every $r \in \Rp$ such that
  $P (h) < r$, then $F (h) \leq r$ for every such $r$, hence
  $F (h) \leq P (h)$; since $h$ is arbitrary, $F \leq P$, so
  $F \in \mathcal C$.

  It remains to verify that $\mathcal C$ is non-empty:
  $\Pred^\rast_\super X$ has a least element by
  Proposition~\ref{prop:PX:sup}~(3), and that least element must be
  below $P$, hence in $\mathcal C$.
\end{proof}

\begin{corollary}
  \label{corl:qusP}
  Let $\rast$ be nothing, ``$\leq 1$'' or ``$1$'' and $X$ be a
  topological space, which we assume compact and non-empty if $\rast$
  is ``$1$''.  Let $\QE \eqdef \Pred^\rast_\super X$,
  $\CF \eqdef \Pred^\rast_\sub X$, and let us define $F^- \perp F^+$ by
  $F^- \leq F^+$.  Then
  $\Img {^\perp\_} \subseteq \CPred^\rast_\super X$.
\end{corollary}
\begin{proof}
  Let $B \subseteq \Pred^\rast_\sub X$.  Then
  ${^\perp B} = \{F^- \in \Pred^\rast_\super X \mid \forall F^+ \in B,
  F^- \leq F^+\}$.  If $B$ is empty, then ${^\perp B}$ is the whole of
  $\Pred^\rast_\super X$, which is the intersection of the empty
  family of sets.
  Since $\Pred^\rast_\super X$ is non-empty (by
  Proposition~\ref{prop:PX:sup}~(3)),
  $\Pred^\rast_\super X \in \CPred^\rast_\super X$.

  If $B$ is non-empty, let $P (h) \eqdef \inf_{F^+ \in B} F^+ (h)$ for
  every $h \in \Lform X$.  $P$ is monotonic and positively homogeneous
  because every $F^+ \in B$ is.  When $\rast$ is ``$1$'', 
  $F^+ (\one)=1$ for every $F^+ \in B$, so (since $B$ is non-empty)
  $P (\one)=1$.  Now
  ${^\perp B} = \{F^- \in \Pred^\rast_\super X \mid F^- \leq P\}$, and
  we conclude by Lemma~\ref{lemma:qusP:0}.
\end{proof}

\begin{lemdef}
  \label{lemdef:qusP}
  Let $\rast$ be nothing, ``$\leq 1$'' or ``$1$'', and $X$ be a
  topological space, which is compact and non-empty if $\rast$ is
  ``$1$''.  There is a map
  $\qusP^\rast_X \colon \Pred^\rast X \to \HV {\Pred^\rast_\super X}$
  defined by
  $\qusP^\rast_X (P) \eqdef \{F \in \Pred^\rast_\super X \mid F \leq
  P\}$, and:
  \begin{enumerate}
  \item $\qusP^\rast_X$ takes its values in $\CPred^\rast_\super X$;
  \item for every $\mathcal C \in \HV {\Pred^\rast_\super X}$,
    $\mathcal C \subseteq \qusP^\rast_X (\supP_X (\mathcal C))$, with
    equality if and only if $\mathcal C \in \CPred^\rast_\super X$;
  \item a set $\mathcal C$ is in $\CPred^\rast_\super X$ if and only if it
    is of the form $\qusP^\rast_X (P)$ for some $P \in \Pred^\rast X$.
  \end{enumerate}
\end{lemdef}
\begin{proof}
  (1) This is a direct consequence of Lemma~\ref{lemma:qusP:0}.

  (2) Let $\mathcal C \in \HV {\Pred^\rast_\super X}$.  For every
  $F \in \mathcal C$, for every $h \in \Lform X$,
  $F (h) \leq \sup_{F \in \mathcal C} F (h) = \supP_X (\mathcal C)
  (h)$, so $F \in \qusP^\rast_X (\supP_X (\mathcal C))$.  Hence
  $\mathcal C \subseteq \qusP^\rast_X (\supP_X (\mathcal C))$.

  If equality holds, 
  then
  $\mathcal C\in \CPred^\rast_\super X$ by~(1).
  Conversely, let
  $\mathcal C \eqdef \bigcap_{i \in I} [h_i \leq 1]^\rast_\super$,
  where $h_i \in \Lform X$, and let $P \eqdef \supP_X (\mathcal C)$.
  For all $F \in \qusP^\rast_X (P)$ and $i \in I$,
  $F (h_i) \leq P (h_i) = \sup_{F \in \mathcal C} F (h_i) \leq 1$, so
  $F \in \mathcal C$.  Therefore
  $\qusP^\rast_X (P) \subseteq \mathcal C$ and hence
  $\mathcal C = \qusP^\rast_X (\supP_X (P))$.

  (3) For every $\mathcal C \in \CPred^\rast_\super X$, we can write
  $\mathcal C$ as $\qusP^\rast_X (P)$ where
  $P \eqdef \supP_X (\mathcal C)$, by~(2).  Conversely, every set of
  the form $\qusP^\rast_X (P)$ where $P \in \Pred^\rast X$ is in
  $\CPred^\rast_\super X$ by~(1).
\end{proof}

Let us investigate the subbasic open subsets of
$\CPred^\rast_\super X$.  We need some additional material from
\cite{Keimel:topcones2}.  A semitopological cone $\C$ is \emph{locally
  convex} if and only if every point has a base of convex open
neighborhoods, and \emph{locally convex-compact} if and only if it has
a base of convex compact neighborhoods.  We will call
\emph{convenient} any topological (not just semitopological) cone that
is locally convex, locally convex-compact, and sober as a topological
space.  (See \cite[Chapter~8]{JGL-topology} for sober spaces.)  In a
convenient cone, every convex compact saturated subset has a base of
convex open neighborhoods \cite[Proposition 10.6]{Keimel:topcones2}.

For every core-compact space $X$, $\Lform X$ is a convenient cone: we
have already said that $\Lform X$ is topological in this case; it
is a continuous dcpo, hence sober \cite[Example 3.5]{GLJ:Valg}, and
locally convex and locally convex-compact
\cite[Lemma~6.12]{Keimel:topcones2}.

We will also need the following technical lemma
\cite[Lemma~6]{JGL:iso:err}.  The segment $[g, h]$ denotes the convex
hull $\conv {\{g, h\}}$.
\begin{lem}
  \label{lemma:Hfriendly:1:V}
  Let $X$ be a topological space such that $\Lform X$ is a convenient
  cone.  Let $g_0 \geq g_1 \geq \cdots \geq g_n \geq \cdots$ be a
  descending sequence of elements of $\Lform X$, and
  $h_0 \in \Lform X$.  Let also $W$ be a Scott-open neighborhood of
  $h_0$ in $\Lform X$ such that $[b \cdot g_n, h]$ is included in $W$
  for every $n \in \nat$, every $b > 1$ and every $h \in W$.  There is
  a convex Scott-open neighborhood $V$ of $h_0$ in $\Lform X$ that
  contains every function $b \cdot g_n$ with $n \in \nat$ and $b > 1$
  and that is included in $W$.
\end{lem}

For every open subset $U$ of $\Pred^\rast_\sub X$, let us abbreviate
$\Diamond U \cap \CPred^\rast_\super X$ as $\Diamond^\pp U$.  For
every natural number $N > n$, we define $\Delta_n [N]$ as the finite
set
$\{\vec b \in \frac 1 N \nat \mid 1 - \frac n N < \sum_{i=1}^n b_i
\leq 1\}$.
\begin{lem}
  \label{lemma:just:Dia:subbase}
  Let $\rast$ be nothing, ``$\leq 1$" or ``$1$", and $X$ be a
  topological space.   If
  $\Lform X$ is a convenient cone, and if $X$ is also compact and
  non-empty when $\rast$ is ``$1$'', then for every $n \geq 1$ and for
  all $h_1, \cdots, h_n \in \Lform X$,
  $\Diamond^\pp {\bigcap_{i=1}^n [h_i > 1]^\rast_\super} = \bigcup_{N
    > n} \bigcap_{\vec b \in \Delta_n [N]} \Diamond^\pp
  {\left[\sum_{i=1}^n b_i h_i > 1\right]^\rast_\super}$.
\end{lem}
\begin{proof}
  For every
  $\mathcal C \in \Diamond^\pp {\bigcap_{i=1}^n [h_i >
    1]^\rast_\super}$, there is an element $F \in \mathcal C$ such
  that for every $i \in \{1, \cdots, n\}$, $F (h_i) > 1$.  Let $N > n$
  be so large that $F (h_i) > \frac N {N-n}$ for every $i$.  By
  superlinearity, for every $\vec a \in \Delta_n$,
  $F (\sum_{i=1}^n a_i h_i) \geq \sum_{i=1}^n a_i F (h_i) > \frac N
  {N-n}$.  For every $\vec b \in \Delta_n [N]$, let
  $\vec a \eqdef \frac 1 {\sum_{i=1}^n b_i} \vec b$: then
  $\vec a \in \Delta_n$, so
  $F (\sum_{i=1}^n b_i h_i) = (\sum_{i=1}^n b_i) F (\sum_{i=1}^n a_i
  h_i) > (\sum_{i=1}^n b_i) \frac N {N-n} \geq (1 - \frac n N) \frac N
  {N-n} = 1$.

  Conversely, let
  $\mathcal C \in \bigcup_{N > n} \bigcap_{\vec b \in \Delta_n [N]}
  \Diamond^\pp {\left[\sum_{i=1}^n b_i h_i > 1\right]^\rast_\super}$, and let
  $P \eqdef \supP_X (\mathcal C)$.  By Lemma and
  Definition~\ref{lemdef:supP}, $P$ is homogeneous and lower
  semicontinuous.  We claim that, for every $\vec a \in \Delta_n$,
  $P (\sum_{i=1}^n a_i h_i) > 1$.  In order to see this, let
  $b_i \eqdef \frac 1 N \lfloor N a_i \rfloor$ for every
  $i \in \{1, \cdots, n\}$; $\vec b$ is in $\Delta_n [N]$, so
  $P (\sum_{i=1}^n b_i h_i) > 1$, and by monotonicity
  $P (\sum_{i=1}^n a_i h_i) > 1$.

  Hence $P$ maps the set
  $\mathcal K \eqdef \conv \{h_1, \cdots, h_n\}$ to $]1, \infty]$, so
  $\mathcal K \subseteq P^{-1} (]1, \infty])$.  $\mathcal K$ is convex
  and compact, because convex hulls of finite unions of convex compact
  sets are \cite[Lemma~4.10]{Keimel:topcones2}.  Since $\Lform X$ is a
  convenient cone, we can find a convex open neighborhood $\mathcal U$
  of $\mathcal K$ that is included in $P^{-1} (]1, \infty])$.
  Additionally, $\mathcal U$ is proper: $0$ is not in
  $P^{-1} (]1, \infty])$, hence not in $\mathcal U$.  It follows that
  $F \eqdef M^{\mathcal U}$ is a superlinear lower semicontinuous map
  from $\Lform X$ to $\creal$, hence a superlinear prevision on $X$.
  Since
  $\mathcal U = F^{-1} (]1, \infty]) \subseteq P^{-1} (]1, \infty])$,
  $F \leq P$.

  This is all we need to do to build $F$ when $\rast$ is nothing.
  When $\rast$ is ``$\leq 1$'', we replace $F$ by a larger
  \emph{subnormalized} superlinear prevision, still below $P$, using
  Corollary~\ref{corl:P:subnorm}.  When $\rast$ is ``$1$'', we use
  Lemma~\ref{lemma:Hfriendly:1:V} with $g_n \eqdef \one$ for each
  $n \in \nat$, and with $\mathcal W \eqdef P^{-1} (]1, \infty])$.
  This lemma applies, because for every $b > 1$ and every
  $h \in \mathcal W$, every element $ab \cdot \one + (1-a) \cdot h$ of
  $[b \cdot \one, h]$ ($a \in [0, 1]$) is mapped by $P$ to
  $ab + (1-a) P (h)$, which is strictly larger than $1$, since $b > 1$
  and $P (h) > 1$.  Hence there is a convex Scott-open neighborhood
  $\mathcal U$ of $\mathcal K$ in $\Lform X$ that contains every
  function $b \cdot \one$ with $b > 1$ and that is included in
  $\mathcal W$.  Using Lemma~\ref{lemma:P:subnorm}~(2), we obtain an
  $F \in \Pred^{\leq 1}_\super X$ such that $F \leq P$ and
  $\mathcal U \subseteq F^{-1} (]1, \infty])$.  In addition,
  $b \cdot \one$ is in $\mathcal U$ for every $b > 1$, so
  $F (b \cdot \one) > 1$, namely $F (\one) > 1/b$ for every $b > 1$.
  Therefore $F (\one) \geq 1$.  Since $F$ is superlinear, for every
  $h \in \Lform X$, $F (\one+h) \geq F (\one) + F (h) \geq 1+F (h)$.
  Hence $F (\one+h) = 1+F(h)$, since $F$ is subnormalized, and
  therefore $F \in \Pred^1_\super X$.

  In any case, $F \in \qusP^\rast_X (P)$.  But
  $\qusP^\rast_X (P) = \qusP^\rast_X (\supP_X (\mathcal C)) = \mathcal
  C$, by Lemma and Definition~\ref{lemdef:qusP}, item~2.  (This is
  where we need to work with $\CPred^\rast_\super X$, not with
  $\HVc {\Pred^\rast_\super X}$.)  Therefore $F \in \mathcal C$.
  Since
  $h_i \in \mathcal K \subseteq \mathcal U \subseteq F^{-1} (]1,
  \infty])$ for every $i \in \{1, \cdots, n\}$, $F (h_i) > 1$.
  Therefore $\mathcal C$ intersects
  $\bigcap_{i=1}^n [h_i > 1]^\rast_\super$ at $F$, showing that
  $\mathcal C \in \Diamond^\pp {\bigcap_{i=1}^n [h_i >
    1]^\rast_\super}$.
\end{proof}

\begin{lem}
  \label{lemma:HDia:lc}
  Let $\rast$ be nothing, ``$\leq 1$'' or ``$1$''.  Let $X$ be a
  space such that $\Lform X$ is a convenient cone; we also
  assume $X$ compact non-empty if $\rast$ is ``$1$''.  The sets
  $\Diamond^\pp {[h > 1]^\rast_\super}$ with $h \in \Lform X$ form a
  subbase of the topology on $\CPred^\rast_\super X$.
\end{lem}
\begin{proof}
  A subbase of the topology on $\Pred^\rast_\super X$ is given by the
  sets $[h > 1]^\rast_\super$ with $h \in \Lform X$.  A base is given
  by finite intersections of those, say
  $\bigcap_{i=1}^n [h_i > 1]^\rast_\super$, with $n \in \nat$ and
  $h_i \in \Lform X$.  Since $\Diamond$ commutes with arbitrary
  unions, a subbase of the topology on $\CPred^\rast_\super X$ is
  given by the sets
  $\Diamond {\bigcap_{i=1}^n [h_i > 1]^\rast_\super} \cap
  \CPred^\rast_\super X = \Diamond^\pp {\bigcap_{i=1}^n [h_i >
    1]^\rast_\super}$.  It remains to rewrite the latter as a union of
  finite intersections of sets of the form
  $\Diamond^\pp {[h > 1]^\rast_\super}$.

  If $n=0$, then
  $\Diamond^\pp {\bigcap_{i=1}^n [h_i > 1]^\rast_\super}$ is the whole
  of $\CPred^\rast_\super X$, which we can write as the empty
  intersection.  If $n \geq 1$, then
  $\Diamond^\pp {\bigcap_{i=1}^n [h_i > 1]^\rast_\super} = \bigcup_{N
    > n} \bigcap_{\vec b \in \Delta_n [N]} \Diamond^\pp
  {\left[\sum_{i=1}^n b_i h_i > 1\right]^\rast_\super}$ by
  Lemma~\ref{lemma:just:Dia:subbase}.
\end{proof}

\begin{lem}
  \label{lemma:qusP:supP}
  Let $\rast$ be nothing, ``$\leq 1$'' or ``$1$''.  Let $X$ be a
  topological space.  For every $P \in \Pred^\rast X$,
  $\supP_X (\qusP^\rast_X (P)) \leq P$, with equality if $\rast$ is
  nothing or ``$\leq 1$'' and $\Lform X$ is locally convex, or if
  $\rast$ is ``$1$'' and $\Lform X$ is a convenient cone.
\end{lem}
\begin{proof}
  For the first part, for all $P \in \Pred X$ and $h \in \Lform X$,
  $\supP_X (\qusP^\rast_X (P)) (h) = \sup_{F \in \Pred^\rast_\super X,
    F \leq P} F (h) \leq P (h)$.

  \emph{Case~1.}  We assume that $\Lform X$ is locally convex.
  
  Let us assume that $\rast$ is nothing.  Let $h \in \Lform X$.  In
  order to see that $P (h) \leq \supP_X (\qusP^\rast_X (P)) (h)$, we
  consider an arbitrary number $r \in \Rp \diff \{0\}$ such that
  $r < P (h)$, and we will show that $r \leq F (h)$ for some
  $F \in \qusP^\rast_X (P)$.  Since $r < P (h)$, $(1/r) \cdot h$ is in
  the open set $P^{-1} (]1, \infty])$.  By local convexity,
  $(1/r) \cdot h$ has a convex open neighborhood $\mathcal U$ that is
  included in $P^{-1} (]1, \infty])$.  Additionally, $\mathcal U$ is
  proper, since the zero function is not in $P^{-1} (]1, \infty])$,
  hence not in $\mathcal U$.  It follows that
  $F \eqdef M^{\mathcal U}$ is a superlinear lower semicontinuous map
  from $\Lform X$ to $\creal$, namely, a superlinear prevision on $X$.
  Since $\mathcal U \subseteq P^{-1} (]1, \infty])$, $F \leq P$, so
  $F \in \qusP_X (P)$.  But $(1/r) \cdot h \in \mathcal U$, so
  $F (h) > r$.

  If $\rast$ is ``$\leq 1$'' instead, we can replace $F$ by a larger
  \emph{subnormalized} superlinear prevision below $P$, using
  Corollary~\ref{corl:P:subnorm}.  Hence, without loss of generality,
  we may assume that $F \in \qusP^{\leq 1}_X (P)$.

  \emph{Case~2.}  We assume that $\rast$ is ``$1$'', and that
  $\Lform X$ is convenient.
  We make use of Lemma~\ref{lemma:Hfriendly:1:V} with
  $g_n \eqdef \one$ for each $n \in \nat$, and with
  $\mathcal K \eqdef \{(1/r) \cdot h\}$ and
  $\mathcal W \eqdef P^{-1} (]1, \infty])$.  The lemma applies: for
  every $b > 1$ and every $h \in \mathcal W$, every element
  $ab \cdot \one + (1-a) \cdot h$ of $\conv {\{b \cdot \one, h\}}$
  ($a \in [0, 1]$) is mapped by $P$ to $ab + (1-a) P (h)$, which is
  strictly larger than $1$, since $b > 1$ and $P (h) > 1$.  Hence
  there is a convex Scott-open neighborhood $\mathcal U$ of
  $\mathcal K$ in $\Lform X$ that contains every function
  $b \cdot \one$ with $b > 1$ and that is included in $\mathcal W$.
  We use Lemma~\ref{lemma:P:subnorm}~(2), and we obtain an
  $F \in \Pred^{\leq 1}_\super X$ such that $F \leq P$ and
  $\mathcal U \subseteq F^{-1} (]1, \infty])$, whence $F (h) > r$.
  Additionally, for every $b > 1$, $b \cdot \one \in \mathcal U$, so
  $F (b \cdot \one) > 1$, namely $F (\one) > 1/b$.  By taking suprema
  over $b$, $F (\one) \geq 1$, and since $F$ is superlinear, for every
  $h \in \Lform X$, $F (\one+h) \geq F (\one) + F (h) \geq 1 + F (h)$.
  The reverse inequality is because $F$ is subnormalized; hence
  $F \in \Pred^1_\super X$.  Since $F \leq P$, $F \in \qusP^1_X (P)$,
  and we remember that $F (h) > r$, as desired.
\end{proof}

\begin{remark}
  \label{rem:qusP:supP:tight}
  When $\rast$ is nothing, the assumption that $\Lform X$ is locally
  convex in Lemma~\ref{lemma:qusP:supP} is tight.  Using the
  one-to-one correspondence between positively homogeneous lower
  semicontinuous functions from $\Lform X$ to $\creal$ and the proper
  (Scott-)open subsets of $\Lform X$, which specializes to a bijection
  between superlinear lower semicontinuous maps and convex proper open
  sets, the equality $\supP_X (\qusP_X (P)) = P$ for every
  $P \in \Pred X$ is equivalent to: every proper open subset is a
  union of convex proper open subsets of $\Lform X$.  We see that the
  latter is equivalent to the requirement that $\Lform X$ be locally
  convex.
\end{remark}

\begin{remark}
  \label{rem:sup:qus}
  Lemma~\ref{lemma:qusP:supP} implies that any prevision in
  $\Pred^\rast X$ can be written as a pointwise supremum of
  superlinear previsions in $\Pred^\rast_\super X$, provided that
  $\Lform X$ is locally convex (or convenient if $\rast$ is ``$1$'').
\end{remark}

\begin{corollary}
  \label{corl:qusP:cont}
  Let $\rast$ be nothing, ``$\leq 1$'' or ``$1$''.  Let $X$ be a
  space such that $\Lform X$ is a convenient cone; we also assume $X$
  compact and non-empty if $\rast$ is ``$1$''.  Then $\qusP^\rast_X$
  is continuous from $\Pred^\rast X$ to $\CPred^\rast_\super X$, hence
  also to $\HV {\Pred^\rast_\super X}$.
\end{corollary}
\begin{proof}
  Using Lemma~\ref{lemma:HDia:lc}, it suffices to show that
  ${(\qusP^\rast_X)}^{-1} (\Diamond^\pp {[h > 1]^\rast_\super})$ is
  open in $\Pred^\rast X$ for every $h \in \Lform X$.  This is the set
  of all previsions $P \in \Pred^\rast X$ such that there is an
  $F \in \Pred^\rast_\super X$ such that $F \leq P$ and $F (h) > 1$.
  Any such $P$ satisfies $P (h) > 1$.  Conversely, every
  $P \in \Pred^\rast X$ such that $P (h) > 1$ can be written as
  $\supP_X (\qusP^\rast_X (P))$ by Lemma~\ref{lemma:qusP:supP}; then
  $\supP_X (\qusP^\rast_X (P)) (h) > 1$, so $F (h) > 1$ for some
  $F \in \qusP^\rast_X (P)$.  Hence
  ${(\qusP_X)}^{-1} (\Diamond^\pp {[h > 1]^\rast_\super}) = [h >
  1]^\rast$.
\end{proof}

We turn to algebraic operations.
\begin{prop}
  \label{prop:Hcvx:cone}
  For every $T_0$ semitopological cone $\C$, $\HVc \C$ is a
  semitopological cone, with zero equal to $\{0\}$, addition defined
  by
  $C_1 + C_2 \eqdef cl \{x+y \mid x \in C_1, \allowbreak y \in C_2\}$,
  and scalar multiplication defined by
  $a \cdot C \eqdef \{a \cdot x \mid x \in C\}$.

  Given any convex subspace $\B$ of $\C$, the function that maps every
  $C \in \HVc \B$ to its closure $cl (C)$ in $\C$ is an affine
  topological embedding, allowing us to see $\HVc \B$ as a convex
  subspace of $\HVc \C$ up to affine homeomorphism.  If $0 \in \B$,
  then $\{0\}$ of $\SVc \C$ belongs to $\SVc \B$ up to that
  homeomorphism, which is then linear.
\end{prop}
\begin{proof}
  The fact that $\HVc \C$ is a cone can be proved just like
  Proposition 4.3 of \cite{TKP:nondet:prob}, although it is only
  stated there for the case where $\C$ is a d-cone.

  Let $U \in \Open \C$.  For every $a \in \Rp$,
  $(a \cdot \_)^{-1} (\Diamond^\cvx U) = \Diamond^\cvx {(a \cdot
    \_)^{-1} (U)}$.  For every $C \in \HVc \C$,
  $(\_ \cdot C)^{-1} (\Diamond^\cvx U) = \bigcup_{x \in C} (\_ \cdot
  x)^{-1} (U)$.  Hence scalar multiplication is separately continuous.
  Fixing $C_0 \in \HVc \C$,
  $(C_0 + \_)^{-1} (\Diamond^\cvx U) = \{C \in \HVc \C \mid U \text{
    intersects }\{x + y \mid x \in C_0, y \in C\}\} = \bigcup_{x \in
    C_0} \Diamond^\cvx {(x + \_)^{-1} (U)}$, so addition is separately
  continuous.

  For what follows, we rely on the following observations: (1) for
  every subset $A$, for every open subset $U$, $U$ intersects $cl (A)$
  if and only if $U$ intersects $A$; (2) two closed subsets are equal
  if and only if they intersect the same open sets.
  
  Given any convex subspace $\B$ of $\C$, let $i$ map every
  $C \in \HVc \B$ to its closure $cl (C)$ in $\C$.  The map $i$ is
  injective: if $i (C_1) = i (C_2)$, then for every open subset
  $U \cap \B$ of $\B$ (where $U$ is open in $\C$), $C_1$ intersects
  $U \cap \B$ if and only if $C_1$ intersects $U$, if and only if
  $i (C_1) = cl (C_1)$ intersects $U$, and working our way backwards
  with $C_2$, if and only if $C_2$ intersects $U \cap \B$.
  Additionally,
  $i^{-1} (\Diamond^\cvx U) = \Diamond^\cvx {(U \cap \B)}$, a subbasic
  open subset of $\HVc \B$.  Conversely, every subbasic open subset of
  $\HVc \B$ is of this form.  Therefore $i$ is a topological
  embedding.  In order to see that $i$ is affine, hence that the image
  of $i$ is convex, we fix $a \in [0, 1]$ and $C_1, C_2 \in \HVc \B$,
  and we will show that $i (C_1) +_a i (C_2) = i (C_1 +_a C_2)$.  For
  every open subset $U$ of $\C$, if $U$ intersects
  $i (C_1) +_a i (C_2)$ then it intersects
  $\{a \cdot x + (1-a) \cdot y \mid x \in cl (C_1), y \in cl (C_2)\}$.
  Let $x \in cl (C_1)$ and $y \in cl (C_2)$ be such that
  $a \cdot x + (1-a) \cdot y \in U$.  Then
  $(a \cdot x + (1-a) \cdot \_)^{-1} (U)$ intersects $cl (C_2)$ at
  $y$, hence also $C_2$.  Let $y' \in C_2$ be such that
  $a \cdot x + (1-a) \cdot y' \in U$.  Then
  $(a \cdot \_ + (1-a) \cdot y')^{-1} (U)$ intersects $cl (C_1)$ at
  $x$, hence also $C_1$.  Hence $U$ intersects
  $\{a \cdot x' + (1-a) \cdot y' \mid x' \in C_1, y' \in C_2\}$, and
  therefore the larger set $i (C_1 +_a C_2)$.  Conversely, if $U$
  intersects $i (C_1 +_a C_2)$, then it intersects $C_1 +_a C_2$,
  which is the closure in $\B$ of
  $\{a \cdot x + (1-a) \cdot y \mid x \in C_1, y \in C_2\}$.
  Therefore $U \cap \B$ intersects the latter closure, so
  $a \cdot x + (1-a) \cdot y \in U \cap \B$ for some $x \in C_1$ and
  $y \in C_2$.  It follows that $U$ intersects $i (C_1) +_a i (C_2)$.

  Finally, if $0$ in $\B$, then $i (\{0\}) = \{0\}$, so $i$ is strict,
  and therefore linear.
\end{proof}


\begin{fact}
  \label{fact:Hcvx:sup}
  Let $\B$ be a convex subspace of a semitopological cone $\C$.  Every
  non-empty family ${(\mathcal C_i)}_{i \in I}$ in $\HVc \B$ has a
  supremum
  $\bigsqcup_{i \in I} \mathcal C_i \eqdef \clconv {\bigcup_{i \in I}
    \mathcal C_i}$.
\end{fact}

\begin{lem}
  \label{lemma:sup:clconv}
  In a convex subspace $\B$ of a semitopological cone $\C$, for every
  subset $A$ of $\B$, $A$, $\conv A$, $cl (A)$ and $\clconv A$ all
  have the same upper bounds in $\B$, with respect to the
  specialization preordering.  In particular, if one has a supremum,
  then so do the other three, and they coincide.
\end{lem}
\begin{proof}
  For every $x \in \B$, if $x$ is an upper bound of $\conv A$, then
  $x$ is an upper bound of the smaller set $A$.  Conversely, if $x$ is
  an upper bound of $A$, then let us consider any element of
  $\conv A$, necessarily of the form $\sum_{i=1}^n a_i \cdot x_i$ for
  some $\vec a \in \Delta_n$, $n \geq 1$ and $x_i \in A$.  Since the
  algebraic operations are continuous, hence monotonic with respect to
  the specialization preordering, the latter is less than or equal to
  $\sum_{i=1}^n a_i \cdot x = x$.  Hence $x$ is an upper bound of
  $\conv A$.

  If $x$ is an upper bound of $cl (A)$, then it is an upper bound of
  $A$.  Conversely, if $x$ is an upper bound of $A$, then
  $A \subseteq \dc x$, and since $\dc x$ is closed,
  $cl (A) \subseteq \dc x$, so that $x$ is an upper bound of $cl (A)$.

  Finally, $\clconv A = cl (\conv A)$ has the same upper bounds as
  $\conv A$, hence as $A$.
\end{proof}

\begin{lem}
  \label{lemma:supP:sup}
  Let $\rast$ be nothing, ``$\leq 1$'' or ``$1$'', and $X$ be a
  topological space.  The map $\supP_X$ preserves non-empty suprema
  from $\HVc {\Pred^\rast_\super X}$ to $\Pred^\rast X$.
\end{lem}
\begin{proof}
  Let ${(\mathcal C_i)}_{i \in I}$ be any non-empty family of elements
  of $\HVc {\Pred^\rast_\super X}$, and
  $\mathcal C \eqdef \bigsqcup_{i \in I} \mathcal C_i$.  For every
  function $h \in \Lform X$,
  $\supP_X (\mathcal C) (h) = \sup_{F \in \bigcup_{i \in I} \mathcal
    C_i} F (h)$ by Lemma~\ref{lemma:sup:clconv}, and this is equal to
  $\sup_{i \in I} \supP_X (\mathcal C_i) (h)$.
\end{proof}

We observe the following.
\begin{fact}
  \label{fact:sup:cl}
  For every space $X$, for every $h \in \Lform X$, for every subset
  $A$ of $X$, $\sup_{x \in A} h (x) = \sup_{x \in cl (A)} h(x)$
  (agreeing that the supremum of an empty set is $0$).
\end{fact}
Indeed, for every $t \in \Rp$, $t < \sup_{x \in cl (A)} h (x)$ if and only
if $h (x) > t$ for some $t \in cl (A)$, if and only if the open set
$h^{-1} (]t, \infty])$ intersects $cl (A)$, if and only if $h^{-1}
(]t, \infty])$ intersects $A$, if and only if $t < \sup_{x \in A} h (x)$.

\begin{lem}
  \label{lemma:supP:+}
  Let $\rast$ be nothing, ``$\leq 1$'' or ``$1$''.  For every
  space $X$, the map $\supP_X$ is affine from
  $\HVc {\Pred^\rast_\super X}$ to $\Pred^\rast X$, and linear if
  $\rast$ is nothing or ``$\leq 1$''.
\end{lem}
\begin{proof}
  For all $\mathcal C_1, \mathcal C_2 \in \HVc {\Pred^\rast_\super X}$ and
  $a \in [0, 1]$, we claim that
  $\supP_X (\mathcal C_1 +_a \mathcal C_2) = \supP_X (\mathcal C_1)
  +_a \supP_X (\mathcal C_2)$.  We recall that
  $\mathcal C_1 +_a \mathcal C_2 = cl \{F_1 +_a F_2 \mid F_1 \in
  \mathcal C_1, F_2 \in \mathcal C_2\}$.  For every $h \in \Lform X$,
  the function $F \in \Pred^\rast_\super X \mapsto F (h)$ is lower
  semicontinuous, so taking the supremum of its values when $F$ varies
  over the closure of a set gives the same result as when $F$ varies
  over the set itself, by Fact~\ref{fact:sup:cl}.  Hence
  $\supP_X (\mathcal C_1 +_a \mathcal C_2) (h) = \sup_{F_1 \in
    \mathcal C_1, F_2 \in \mathcal C_2} (a\, F_1 (h) + (1-a)\, F_2
  (h)) = a\, \sup_{F_1 \in \mathcal C_1} F_1 (h) + (1-a)\,\sup_{F_2
    \in \mathcal C_2} F_2 (h) = a\, \supP_X (\mathcal C_1) (h) +
  (1-a)\, \supP_X (\mathcal C_2) (h)$.  It follows that $\supP_X$ is
  affine.

  When $\rast$ is nothing or ``$\leq 1$'', the least element of
  $\HVc {\Pred^\rast_\super X}$ is $\{0\}$, where we write $0$ for the
  zero prevision, and $\supP_X \{0\} (h) = 0$ for every
  $h \in \Lform X$.  Therefore $\supP_X$ is strict, hence
  linear.
\end{proof}

\begin{thm}
  \label{thm:sup:qus:retr}
  Let $\rast$ be nothing, ``$\leq 1$'' or ``$1$''.  For every space
  $X$, $\Pred^\rast X$ is a retract of $\HVc {\Pred^\rast_\super X}$
  through $\supP_X$ and $\qusP^\rast_X$; the retraction $\supP_X$ is affine
  (linear if $\rast$ is nothing or ``$\leq 1$'') and preserves suprema
  of non-empty families.
\end{thm}
\begin{proof}
  The map $\supP_X$ preserves non-empty suprema by
  Lemma~\ref{lemma:supP:sup}, is affine (resp.\ linear) by
  Lemma~\ref{lemma:supP:+} and continuous by Lemma and
  Definition~\ref{lemdef:supP}; $\qusP^\rast_X$ is continuous by
  Corollary~\ref{corl:qusP:cont}, and a right inverse to $\supP_X$ by
  Lemma~\ref{lemma:qusP:supP}.
\end{proof}

\begin{thm}[Second isomorphism]
  \label{thm:sup:qus:iso}
  Let $\rast$ be nothing, ``$\leq 1$'' or ``$1$''.  Let $X$ be a
  topological space such that $\Lform X$ is a convenient cone, for
  example a core-compact space; we also assume $X$ compact and
  non-empty if $\rast$ is ``$1$''.  Then $\CPred^\rast_\super X$ and
  $\Pred^\rast X$ are homeomorphic through $\supP_X$ and
  $\qusP^\rast_X$, hence isomorphic 
\end{thm}
\begin{proof}
  By Lemma and Definition~\ref{lemdef:supP}, $\supP_X$ defines (by
  restriction) a continuous map from $\CPred^\rast_\super X$ to
  $\Pred^\rast X$; $\qusP^\rast_X$ is continuous by
  Corollary~\ref{corl:qusP:cont}.  They are mutual inverses by Lemma
  and Definition~\ref{lemdef:qusP}, item~2, and by
  Lemma~\ref{lemma:qusP:supP}.
\end{proof}

\begin{remark}
  \label{rem:second:double:AN}
  We recall the homeomorphism
  $\Pred^\rast_{\super} X \cong \SVc {\Pred^\rast_{\lin} X}$
  \cite[Theorem~4.15]{JGL-mscs16}.  Then Theorem~\ref{thm:sup:qus:iso}
  expresses $\Pred^\rast X$ as a kind of double hyperspace over
  $\Pred^\rast_\lin X$ ($\cong \Val_\rast X$), with the Hoare and
  Smyth constructions in the opposite order, compared to
  Remark~\ref{rem:first:double:DN}.
\end{remark}

\begin{remark}
  \label{rem:HBox}
  $\CPred^\rast_\super X$ is included in $\HVc {\Pred^\rast_\super X}$
  by definition.  The inclusion is strict in general: any of the
  following two examples are counterexamples.
\end{remark}


\begin{example}
  \label{ex:HDia:sup}
  $\Pred^\rast X$ has all non-empty suprema, hence so has
  $\CPred^\rast_\super X$.  But they differ from suprema
  $\mathcal C_1 \sqcup \mathcal C_2 \eqdef \clconv {(\mathcal C_1 \cup
    \mathcal C_2)}$ computed in $\HVc {\Pred^\rast_\super X}$ in
  general.  Let $X \eqdef \{0, 1\}$ with the discrete topology,
  $\Lambda (h) \eqdef \frac 1 2 (h (0) + h (1))$,
  $P_1 (h) \eqdef \min (h (0), \Lambda (h))$,
  $P_2 (h) \eqdef \min (h (1), \Lambda (h))$ for every
  $h \in \Lform X$.  Then, imitating Example~\ref{ex:QBox:inf}:
  \begin{enumerate}
  \item $P_1$, $P_2$ are normalized (superlinear) previsions on $X$
    and $\Lambda$ is the supremum of $P_1$ and $P_2$ in
    $\Pred^\rast X$.
  \item $\Lambda$ is not in
    $\clconv {(\qusP^\rast_X (P_1) \cup \qusP^\rast_X (P_2))}$.  We
    proceed with the same $\underline F$ notation as previously: for
    every $t \in [0, 1]$, $\underline \Lambda (t) = \frac 1 2$,
    $\underline P_1 (t) = \min (t, \frac 1 2)$, and
    $\underline P_2 (t) = \min (1-t, \frac 1 2)$.  Let us assume that
    $\Lambda \in \clconv {(\qusP^\rast_X (P_1) \cup \qusP^\rast_X
      (P_2))}$.  Every open neighborhood $\mathcal U$ of $\Lambda$ in
    $\Pred^\rast_\super X$ must intersect
    $\conv (\qusP^\rast_X (P_1) \cup \qusP^\rast_X (P_2))$.  For
    $\mathcal U$, we pick an arbitrary, but fixed number $\epsilon$ in
    $]0, \frac 1 2[$, we fix $r \in [(1-\epsilon)/2, 1/2[$, and we
    take the set
    $[h_{\frac 1 2 +\epsilon} > r] \cap [h_{\frac 1 2 -\epsilon} >
    r]$.  (See Example~\ref{ex:QBox:inf}~(2) for $h_t$.)  For every
    $t \in [0, 1]$, $\Lambda (h_t) = \frac 1 2$, so
    $\Lambda \in \mathcal U$.  Hence some element
    $a\, F_1 + (1-a)\, F_2$ of
    $\conv {(\qusP^\rast_X (P_1) \cup \qusP^\rast_X (P_2))}$, with
    $a \in [0, 1]$ and $F_1 \leq P_1$, $F_2 \leq P_2$, must be in
    $\mathcal U$.  This entails that
    $a\, P_1 (h_{\frac 1 2 +\epsilon}) + (1-a)\, P_2 (h_{\frac 1 2
      +\epsilon}) > r$ and
    $a\, P_1 (h_{\frac 1 2 -\epsilon}) + (1-a)\, P_2 (h_{\frac 1 2
      -\epsilon}) > r$, in other words
    $a\, \underline P_1 (\frac 1 2 +\epsilon) + (1-a)\, \underline P_2
    (\frac 1 2 +\epsilon) > r$ and
    $a\, \underline P_1 (\frac 1 2 -\epsilon) + (1-a)\, \underline P_2
    (\frac 1 2 -\epsilon) > r$.  Since $\epsilon > 0$,
    $\underline P_1 (\frac 1 2 +\epsilon) = \frac 1 2$ and
    $\underline P_2 (\frac 1 2 +\epsilon) = \frac 1 2 - \epsilon$, so
    the first inequality means
    $a \frac 1 2 + (1-a) (\frac 1 2 - \epsilon) > r$, namely
    $\frac 1 2 - (1-a) \epsilon > r$.  Similarly,
    $\underline P_1 (\frac 1 2 -\epsilon) = \frac 1 2 - \epsilon$ and
    $\underline P_2 (\frac 1 2 - \epsilon) = \frac 1 2$, so the second
    inequality means $a (\frac 1 2 - \epsilon) + (1-a) \frac 1 2 > r$,
    namely $\frac 1 2 - a \epsilon > r$.  Adding the two, we obtain
    that $1 - \epsilon > 2r$.  This is impossible since
    $r \in [(1-\epsilon)/2, 1/2[$.
  \item It follows that
    $\qusP^\rast_X (\sup (P_1, P_2)) \neq \clconv {(\qusP^\rast_X
      (P_1) \cup \qusP^\rast_X (P_2))}$, hence that $\qusP^\rast_X$
    does not map binary suprema in $\Pred^\rast X$ to binary suprema in
    $\HVc {\Pred^\rast_\super X}$.
  \item It also follows that binary suprema
    $\mathcal C_1 \sqcup \mathcal C_2$ in
    $\HVc {\Pred^\rast_\super X}$ of elements of
    $\CPred^\rast_\super X$ may fail to be in $\CPred^\rast_\super X$,
    and in particular $\CPred^\rast_\super X$ is a proper subspace of
    $\HVc {\Pred^\rast_\super X}$.  Indeed, let
    $\mathcal C_1 \eqdef \qusP^\rast_X (P_1)$ and
    $\mathcal C_2 \eqdef \qusP^\rast_X (P_2)$.  Then
    $\mathcal C_1 \sqcup \mathcal C_2$ cannot be in
    $\CPred^\rast_\super X$, otherwise it would be the supremum of
    $\mathcal C_1$ and $\mathcal C_2$ in $\CPred^\rast_\super X$.
    However, that supremum is
    $\qusP^\rast_X (\sup (\supP_X (\mathcal C_1), \allowbreak \supP_X
    (\mathcal C_2)) = \qusP^\rast_X (\sup (P_1, P_2))$, which is
    different by~(4).  The latter equality is justified by
    Theorem~\ref{thm:sup:qus:iso}, which applies since $X$ is discrete
    and finite, hence trivially locally compact, so $\Lform X$ is a
    convenient cone.  \qed
  \end{enumerate}
\end{example}

\begin{example}
  \label{ex:HDia:+}
  Transporting the cone structure of $\Pred X$ over to
  $\CPred_\super X$, the latter is also a cone.  As in
  Example~\ref{ex:QBox:+}, we show that addition can differ from
  addition in the enclosing cone $\HVc {\Pred_\super X}$, and
  $\qusP_X$ is not in general linear.  We take $X \eqdef \{0, 1\}$
  once again, $\Lambda (h) \eqdef h (0)+h(1)$,
  $P_1 (h) \eqdef \max (h (0), h (1))$ and
  $P_2 (h) \eqdef \min (2 h (0) + h (1), h (0) + 2 h (1))$, exactly as
  in Example~\ref{ex:QBox:+}.  Then:
  \begin{enumerate}
  \item For every superlinear prevision $F_1 \leq P_1$, there is a
    linear prevision $\Lambda_1$ such that
    $F_1 \leq \Lambda_1 \leq P_1$, by Keimel's sandwich theorem.  Then
    there are numbers $a, b \in \creal$ such that for every
    $t \in [0, 1]$, $\underline \Lambda_1 (t) = a t+ b (1-t)$, and
    then $a + b \leq 1$: this is because
    $\underline \Lambda_1 \leq \underline P_1$, evaluated at
    $t \eqdef \frac 1 2$.
  \item In the situation of~(1), it is impossible that
    $F_1 +_{\frac 1 2} P_2 \in [h_0 > 1 - \epsilon] \cap [h_1 > 1 -
    \epsilon]$ for any $\epsilon \in {]0, \frac 1 4[}$.  Indeed, since
    $F_1 \leq \Lambda_1$, $\Lambda_1 +_{\frac 1 2} P_2$ would also be
    in $[h_0 > 1 - \epsilon] \cap [h_1 > 1 - \epsilon]$.  Then
    $\underline \Lambda_1 (0) + \underline P_2 (0) > 2 - 2\epsilon$
    and
    $\underline \Lambda_1 (1) + \underline P_2 (1) > 2 - 2\epsilon$.
    Since $\underline P_2 (t) = \min (1+t, 2-t)$ for every
    $t \in [0, 1]$, we get that $b + 1 > 2 - 2\epsilon$ and
    $a + 1 > 2 - 2\epsilon$.  Summing the two,
    $a + b > 2 - 4\epsilon$; this is
    impossible since $a+b \leq 1$ and $\epsilon < \frac 1 4$.
  \item 
    For every $\epsilon \in {]0, \frac 1 4[}$,
    $\mathcal U \eqdef [h_0 > 1-\epsilon] \cap [h_1 > 1-\epsilon]$
    intersects $\qusP_X (P_1 +_{\frac 1 2} P_2)$ at $\Lambda$ (indeed,
    $P_1 +_{\frac 1 2} P_2 = \Lambda$), but does not intersect
    $\qusP_X (P_1) +_{\frac 1 2} \qusP_X (P_2)$.  If it did,
    $\mathcal U$ would intersect
    $\{F_1 +_{\frac 1 2} F_2 \mid F_1 \in \qusP_X (P_1), F_2 \in
    \qusP_X (P_2)\}$, so there would be superlinear previsions
    $F_1 \leq P_1$ and $F_2 \leq P_2$ such that
    $\frac 1 2 F_1 + \frac 1 2 F_2 \in \mathcal U$.  Since
    $F_2 \leq P_2$, in particular
    $\frac 1 2 F_1 + \frac 1 2 P_2 \in \mathcal U$, contradicting~(2).
  \item It follows that $\qusP_X$ is not affine from $\Pred X$ to
    $\HVc {\Pred_\super X}$: by (3),
    $\qusP_X (P_1 +_{\frac 1 2} P_2) \neq \qusP_X (P_1) +_{\frac 1
      2}\qusP_X (P_2)$.
  \item Convex combinations of elements of $\CPred_\super X$, as
    computed in $\HVc {\Pred_\super X}$, may fail to be in
    $\CPred_\super X$.  In order to see this, let
    $\mathcal C_1 \eqdef \qusP_X (P_1)$ and
    $\mathcal C_2 \eqdef \qusP_X (P_2)$.  Then
    $\mathcal C_1 +_{\frac 1 2} \mathcal C_2 \neq \qusP_X (P_1
    +_{\frac 1 2} P_2)$ by (4).  But, if
    $\mathcal C_1 +_{\frac 1 2} \mathcal C_2$ were in
    $\CPred_\super X$, we would have
    $\mathcal C_1 +_{\frac 1 2} \mathcal C_2 = \qusP_X (\supP_X
    (\mathcal C_1 +_{\frac 1 2} \mathcal C_2))$ by
    Lemma~\ref{lemma:qusP:supP}; the latter applies because $X$ is
    locally compact, so $\Lform X$ is a convenient cone.  Now
    $\qusP_X (\supP_X (\mathcal C_1 +_{\frac 1 2} \mathcal C_2)) =
    \qusP_X (\supP_X (\mathcal C_1) +_{\frac 1 2} \supP_X (\mathcal
    C_2))$ (by Lemma~\ref{lemma:supP:+})
    $= \qusP_X (P_1 +_{\frac 1 2} P_2)$ (by
    Lemma~\ref{lemma:qusP:supP}).  \qed
  \end{enumerate}
\end{example}

\section{Back to orthogonality relations}
\label{sec:back-orth-relat}

The material of Section~\ref{sec:PX=QPAPX} and
Section~\ref{sec:PX=CPDPX} allows us to conclude that
$\CPred^\rast_\super X$ and $\QPred^\rast_\sub X$ do arise from the
orthogonality constructions of Section~\ref{sec:double-hypersp-orth},
and relates orthogonals with the maps $\nimP^\rast_X \circ \supP_X$
and $\qusP^\rast_X \circ \minP_X$, as we show.


\begin{thm}
  \label{thm:HppQ}
  Let $\rast$ be nothing, ``$\leq 1$'' or ``$1$''.  Let $X$ be a
  topological space such that $\Lform X$ is a convenient cone; we also
  assume that $X$ is compact and non-empty when $\rast$ is ``$1$''.
  Let $\QE \eqdef \Pred^\rast_\super X$,
  $\CF \eqdef \Pred^\rast_\sub X$, and let us define $F^- \perp F^+$
  by $F^- \leq F^+$.  Then:
  \begin{enumerate}
  \item
    $(\nimP^\rast_X \circ \supP_X) (\mathcal C) = {\mathcal C}^\perp$
    for every $\mathcal C \in \CF$;
  \item
    $(\qusP^\rast_X \circ \minP_X) (\mathcal Q) = {^\perp \mathcal Q}$
    for every $\mathcal Q \in \QE$;
  \item $\Img {\_^\perp} = \QPred^\rast_\sub X$;
  \item $\Img {^\perp\_} = \CPred^\rast_\super X$.
  \end{enumerate}
\end{thm}
\begin{proof}
  1.
  $(\nimP^\rast_X \circ \supP_X) (\mathcal C) = \{F^+ \in
  \Pred^\rast_\sub X \mid \supP_X (\mathcal C) \leq F^+\} = \{F^+ \in
  \Pred^\rast_\sub X \mid \forall F^- \in \mathcal C, F^- \leq F^+\} =
  {\mathcal C}^\perp$.

  2.
  $(\qusP^\rast_X \circ \minP_X) (\mathcal Q) = \{F^- \in
  \Pred^\rast_\super X \mid F^- \leq \minP_X (\mathcal Q)\} = \{F^- \in
  \Pred^\rast_\super X \mid \forall F^+ \in \mathcal Q, F^- \leq F^+\}
  = {^\perp \mathcal Q}$.

  3.  By Corollary~\ref{corl:nimP},
  $\Img {\_^\perp} \subseteq \QPred^\rast_\sub X$.  Conversely, every
  element $\mathcal Q$ of $\QPred^\rast_\sub X$ can be written as
  $\nimP^\rast_X (\supP_X (\qusP^\rast_X (\minP_X (\mathcal Q))))$
  using the homeomorphisms of Theorem~\ref{thm:sup:qus:iso} and
  Theorem~\ref{thm:min:nim:iso}; hence as
  $({^\perp \mathcal Q})^\perp$, by items~1 and~2, and that falls into
  $\Img \_^\perp$.

  4.  By Corollary~\ref{corl:qusP},
  $\Img {^\perp\_} \subseteq \CPred^\rast_\super X$.  Conversely, every
  element $\mathcal C$ of $\CPred^\rast_\super X$ can be written as
  $\qusP^\rast_X (\minP_X (\nimP^\rast_X (\supP_X (\mathcal C))))$ using
  the homeomorphisms of Theorem~\ref{thm:sup:qus:iso} and
  Theorem~\ref{thm:min:nim:iso}; hence as
  ${^\perp ({\mathcal C}^\perp)}$, by items~1 and~2, and that is in
  $\Img {^\perp \_}$.
\end{proof}



\bibliographystyle{abbrv}
\DeclareRobustCommand{\VAN}[3]{#3}
\bibliography{justprev}

@string{ieee = 	"IEEE Computer Society Press"}

@inproceedings{Gou-csl07,
  	address =	{Lausanne, Switzerland},
  	author =	{Goubault{-}Larrecq, Jean},
  	booktitle =	{{P}roceedings of the 16th {A}nnual {EACSL} {C}onference on {C}omputer {S}cience {L}ogic ({CSL}'07)},
  	DOI =	{10.1007/978-3-540-74915-8_40},
  	editor =	{Duparc, Jacques and Henzinger, {\relax Th}omas A.},
  	month =	sep,
  	pages =	{542-557},
  	publisher =	{Springer},
  	series =	{Lecture Notes in Computer Science},
  	title =	{Continuous Previsions},
  	url =	{http://www.lsv.ens-cachan.fr/Publis/PAPERS/PDF/JGL-csl07.pdf},
  	volume =	{4646},
  	year =	{2007},
}

@article{JGL-mscs16,
  publisher = {Cambridge University Press},
  journal = {Mathematical Structures in Computer Science},
  author = {Goubault{-}Larrecq, Jean},
  title = {Isomorphism Theorems between Models of Mixed Choice},
  volume = {27},
  number = {6},
  pages = {1032-1067},
  month = sep,
  year = 2017,
  url = {http://www.lsv.fr/Publis/PAPERS/PDF/JGL-mscs16.pdf},
  doi = {10.1017/S0960129515000547},
    note = {Revised version on arXiv:2411.13500 [cs.LO]}
}

@article{JGL:iso:err,
  title={Errata to ``Isomorphism theorems between models of mixed choice,'' fixes and consequences},
  volume={35},
  DOI={10.1017/S0960129525100200},
  journal={Mathematical Structures in Computer Science},
  author={Goubault{-}Larrecq, Jean},
  year={2025},
  pages={e23}
}

@article{dBK:comm,
  TITLE = {On the Commutativity of the Powerspace Constructions},
  AUTHOR = {de Brecht, Matthew and Kawai, Tatsuji},
  URL = {https://lmcs.episciences.org/5679},
  DOI = {10.23638/LMCS-15(3:13)2019},
  JOURNAL = {Logical Methods in Computer Science},
  VOLUME = {15},
  NUMBER = {3},
  YEAR = {2019},
  OPTMONTH = Aug,
}

@book{JGL-topology,
   alias = {Topology},
   optnote = {By Jean Goubault{-}Larrecq. ISBN 978-1-107-03413-6},
   author = {Goubault{-}Larrecq, Jean},
   publisher = {Cambridge University Press}, 
   title = {Non-{H}ausdorff Topology and Domain Theory---Selected Topics in Point-Set Topology},
   series = {New Mathematical Monographs},
   volume = {22},
   year = {2013},
   OPTnote = {ISBN 978-1-107-03413-6},
}

@InProceedings{jones89,
  author =	 {Jones, Claire and Plotkin, Gordon},
  title =	 {A Probabilistic Powerdomain of Evaluations},
  pages =	 {186--195},
  booktitle =	 {Proceedings of the 4th Annual Symposium on Logic in
                  Computer Science},
  year =	 1989,
  publisher =	 ieee,
  category =	 {inproc},
  source =	 {Abramsky and Jung, 1994},
}

@PhdThesis{Jones:proba,
  author = 	 {Jones, Claire},
  title = 	 {Probabilistic Non-Determinism},
  school = 	 {University of Edinburgh},
  year = 	 {1990},
  OPTkey = 	 {},
  OPTtype = 	 {},
  OPTaddress = 	 {},
  OPTmonth = 	 {},
  note =	 {Technical Report ECS-LFCS-90-105},
  OPTannote = 	 {}
}

@Article{saheb-djahromi:meas,
  author = 	 {Saheb-Djahromi, Nait},
  title = 	 {Cpo's of Measures for Nondeterminism},
  journal = 	 {Theoretical Computer Science},
  year = 	 {1980},
  OPTkey = 	 {},
  volume =	 {12},
  OPTnumber = 	 {},
  pages =	 {19--37},
  OPTmonth = 	 {},
  OPTnote = 	 {},
  OPTannote = 	 {}
}

@InProceedings{Lawson:valuation,
  author = 	 {Lawson, Jimmie Don},
  title = 	 {Valuations on Continuous Lattices},
  OPTcrossref =  {},
  OPTkey = 	 {},
  booktitle =	 {Mathematische Arbeitspapiere},
  year =	 {1982},
  editor =	 {R.-E. Hoffmann},
  volume =	 {27},
  OPTnumber = 	 {},
  OPTseries = 	 {},
  pages =	 {204--225},
  OPTmonth = 	 {},
  address =	 {Universit{\"a}t Bremen},
  OPTorganization = {},
  OPTpublisher = {},
  OPTnote = 	 {},
  OPTannote = 	 {}
}

@Article{Adamski:measures,
  author = 	 {Adamski, Wolfgang},
  title = 	 {{$\tau$}-Smooth {B}orel Measures on Topological Spaces},
  journal = 	 {Mathematische Nachrichten},
  year = 	 {1977},
  OPTkey = 	 {},
  volume =	 {78},
  OPTnumber = 	 {},
  pages =	 {97--107},
  OPTmonth = 	 {},
  OPTnote = 	 {},
  OPTannote = 	 {}
}

@Article{dBGLJL:LCS,
  author = 	 {de Brecht, Matthew and Goubault{-}Larrecq, Jean and Jia, Xiaodong and Lyu, Zhenchao},
  title = 	 {Domain-Complete and {LCS}-Complete Spaces},
  journal = 	 {Electronic Notes in Theoretical Computer Science},
  publisher = {Elsevier Science Publishers B. V.},
  address = {Amsterdam, The Netherlands, The Netherlands},
  year = 	 {2019},
  OPTkey = 	 {},
  volume =	 {345},
  OPTnumber = 	 {},
  pages =	 {3--35},
  OPTmonth = 	 aug # {28},
  note =	 {Proc. 8th International Symposium on Domain Theory (ISDT'19)},
  OPTannote = 	 {}
}

@MastersThesis{Tix:bewertung,
  author = 	 {Tix, Regina},
  title = 	 {{S}tetige {B}ewertungen auf topologischen {R\"a}umen},
  school =	 {Technische Hochschule Darmstadt},
  year = 	 {1995},
  OPTkey = 	 {},
  type =	 {Diplomarbeit},
  OPTaddress = 	 {},
  OPTmonth =	 jun,
  OPTnote = 	 {},
  OPTannote = 	 {}
}

@InProceedings{Jung:scs:prob,
  author = 	 {Jung, Achim},
  title = 	 {Stably Compact Spaces and the Probabilistic Powerspace Construction},
  booktitle =	 {Domain-theoretic Methods in Probabilistic Processes},
  year =	 2004,
  editor =	 {J. Desharnais and P. Panangaden},
  volume =	 87,
  pages = {5--20},
  series =	 {Electronic Lecture Notes in Computer Science},
  publisher = {Elsevier Science Publishers B. V.},
  address = {Amsterdam, The Netherlands, The Netherlands},
  note =	 {15pp.}
}

@Book{Walley:prev,
  author =	 {Walley, Peter},
  ALTeditor = 	 {},
  title = 	 {Statistical Reasoning with Imprecise Probabilities},
  publisher = 	 {Chapman and Hall},
  year = 	 {1991},
  OPTkey = 	 {},
  OPTvolume = 	 {},
  OPTnumber = 	 {},
  OPTseries = 	 {},
  address =	 {London},
  OPTedition = 	 {},
  OPTmonth = 	 {},
  OPTnote = 	 {},
  OPTannote = 	 {}
}

@Article{Keimel:topcones2,
  author = 	 {Keimel, Klaus},
  title = 	 {Topological Cones: Functional Analysis in a {$T_0$}-Setting},
  journal = 	 {Semigroup Forum},
  year = 	 {2008},
  OPTkey = 	 {},
  volume =	 {77},
  number =	 {1},
  pages =	 {109--142},
  OPTmonth = 	 may,
  OPTnote = 	 {},
  OPTannote = 	 {}
}

@article{VT:double:power,
 author = {Vickers, Steven J. and Townsend, Christopher F.},
 title = {A Universal Characterization of the Double Powerlocale},
 journal = {Theoretical Computer Science},
 volume = {316},
 number = {1-3},
 year = {2004},
 issn = {0304-3975},
 pages = {297--321},
 doi = {10.1016/j.tcs.2004.01.034},
 publisher = {Elsevier Science Publishers Ltd.},
 address = {Essex, UK},
 }

@Article{FM:QH,
  author = 	 {Flannery, Kevin E. and Martin, Johannes J.},
  title = 	 {The {H}oare and {S}myth Power Domain Constructors Commute under Composition},
  journal = 	 {Journal of Computer and System Sciences},
  year = 	 {1990},
  OPTkey = 	 {},
  volume =	 {40},
  OPTnumber = 	 {},
  pages =	 {125--135},
  OPTmonth = 	 {},
  OPTnote = 	 {},
  OPTannote = 	 {}
}

@Article{Heckmann:QH,
  author = 	 {Heckmann, Reinhold},
  title = 	 {Lower and Upper Power Domain Constructions Commute on all Dcpos},
  journal = 	 {Information Processing Letters},
  year = 	 {1991},
  OPTkey = 	 {},
  volume =	 {40},
  OPTnumber = 	 {},
  pages =	 {7--11},
  OPTmonth = 	 {},
  OPTnote = 	 {},
  OPTannote = 	 {}
}

@Book{Birkhoff,
  author =	 {Birkhoff, Garrett},
  ALTeditor = 	 {},
  title = 	 {Lattice Theory},
  publisher = 	 {American Mathematical Society},
  year = 	 {1940},
  OPTkey = 	 {},
  OPTvolume = 	 {},
  OPTnumber = 	 {},
  OPTseries = 	 {},
  OPTaddress = 	 {},
  OPTedition = 	 {},
  OPTmonth = 	 {},
  note =	 {Second edition, 1948.  Third edition, 1967.},
  OPTannote = 	 {}
}

@article{GLK-mscs10,
  author =        {Goubault{-}Larrecq, Jean and Keimel, Klaus},
  journal =       {Mathematical Structures in Computer Science},
  month =         jun,
  number =        {3},
  pages =         {511-561},
  publisher =     {Cambridge University Press},
  title =         {{C}hoquet-{K}endall-{M}atheron Theorems for
                   Non-{H}ausdorff Spaces},
  volume =        {21},
  year =          {2011},
  doi =           {10.1017/S0960129510000617},
  url =           {http://www.lsv.ens-cachan.fr/Publis/PAPERS/PDF/
                   GLK-mscs10.pdf},
  url-lsv =       {http://www.lsv.ens-cachan.fr/Publis/
                  publis.php?onlykey=GLK-mscs10},
}

@Article{TKP:nondet:prob,
  author = 	 {Tix, Regina and Keimel, Klaus and Plotkin, Gordon},
  title = 	 {Semantic Domains for Combining Probability and Non-Determinism},
  journal = 	 {Electronic Notes in Theoretical Computer Science},
  year = 	 {2009},
  OPTkey = 	 {},
  volume =	 {222},
  OPTnumber = 	 {},
  pages =	 {3--99},
  OPTmonth = 	 {feb},
  note =	 {Originally published in 2005.},
  publisher = {Elsevier Science Publishers B. V.},
  address = {Amsterdam, The Netherlands, The Netherlands},
  OPTannote = 	 {}
}

@Article{GLJ:Valg,
  author = 	 {Goubault{-}Larrecq, Jean and Jia, Xiaodong},
  title = 	 {Algebras of the Extended Probabilistic Powerdomain Monad},
  journal = 	 {Electronic Notes in Theoretical Computer Science},
  year = 	 {2019},
  OPTkey = 	 {},
  volume =	 {345},
  OPTnumber = 	 {},
  pages =	 {37--61},
  OPTmonth = 	 aug # {28},
  note =	 {Proc. 8th International Symposium on Domain Theory (ISDT'19)},
  publisher = {Elsevier Science Publishers B. V.},
  address = {Amsterdam, The Netherlands, The Netherlands},
  OPTannote = 	 {}
}

\end{document}